\documentclass[a4paper,UKenglish,cleveref, autoref, thm-restate]{lipics-v2021}

\def\N{\mathbb{N}}

\def\P{\mathcal{P}}

\def\B{\square}
\def\T{\mathcal{T}}

\def\E{\mathcal{E}}

\def\eps{\varepsilon}

\def\dist{\text{dist}}

\def\Nei{\text{Nei}}
\def\coord{\text{coord}}
\def\loc{\text{loc}}
\def\type{\text{type}}

\def\minPts{\mathit{minPts}}

\newcommand{\tony}[1]{}
\renewcommand{\tony}[1]{{\color{blue}{\bf{Tony says:}} \emph{#1}}}

\newcommand{\zhuo}[1]{}
\renewcommand{\zhuo}[1]{{\color{red}{\bf{zhuo:}} \emph{#1}}}

\newcommand{\junhao}[1]{}
\renewcommand{\junhao}[1]{{\color{purple}{\bf{Junhao says:}} \emph{#1}}}

\newcommand{\TODO}[1]{}
\renewcommand{\TODO}[1]{{\color{green}{\bf{TODO:}} \emph{#1}}}

\definecolor{orange}{RGB}{255, 165, 0}
\definecolor{limegreen}{RGB}{50, 205, 50}
\definecolor{brightyellow}{RGB}{255, 215, 0}

\newcommand{\bfrev}[1]{\iffalse #1 \fi}

\usepackage{amsmath}
\usepackage{balance}
\usepackage{graphicx}
\usepackage{subcaption}
\usepackage{hyperref}
\usepackage{wrapfig}
\usepackage{mathrsfs}
\usepackage{mdframed}
\usepackage{microtype} 
\usepackage{multirow}
\usepackage[skins,breakable]{tcolorbox}
\usepackage[normalem]{ulem}
\usepackage{xspace}
\usepackage{tikz}

\usepackage{ulem}
\usepackage[linesnumbered,ruled,vlined]{algorithm2e}

\SetAlgoNoEnd

\newtheorem{fact}{Fact}

\title{$O(1)$-Round MPC Algorithms for Multi-dimensional Grid Graph Connectivity, Euclidean MST and DBSCAN} 

\titlerunning{$O(1)$-Round MPC Algorithms for Grid Graph Connectivity, EMST and DBSCAN} 

\author{Junhao Gan }{The University of Melbourne, Australia}{junhao.gan@unimelb.edu.au}{}{}

\author{Anthony Wirth}{The University of Melbourne, Australia}{anthony.wirth@sydney.edu.au}{}{}

\author{Zhuo Zhang}{The University of Melbourne, Australia}{zhuo.zhang@student.unimelb.edu.au}{}{}

\authorrunning{J. Gan, A. Wirth, and Z. Zhang} 

\Copyright{Junhao Gan, Anthony Wirth, and Zhuo Zhang} 



\acknowledgements{This work is in part supported by ARC DE190101118 and DP230102908.}

\EventEditors{Sudeepa Roy and Ahmet Kara}
\EventNoEds{2}
\EventLongTitle{28th International Conference on Database Theory (ICDT 2025)}
\EventShortTitle{ICDT 2025}
\EventAcronym{ICDT}
\EventYear{2025}
\EventDate{March 25--28, 2025}
\EventLocation{Barcelona, Spain}
\EventLogo{}
\SeriesVolume{328}
\ArticleNo{4}

\usepackage[bibliography=common]{apxproof}

\AtEndPreamble{
\newtheoremrep{claimapp}[theorem]{Claim}
\newtheoremrep{theoremapp}[theorem]{Theorem}
\newtheoremrep{corollaryapp}[corollary]{Corollary}
\newtheoremrep{lemmaapp}[theorem]{Lemma}
\renewcommand{\mainbodyrepeatedtheorem}{\emph{\textbf{[*]}}}
}

\begin{document}
\nolinenumbers

\maketitle

\begin{abstract}
In this paper, we investigate three fundamental problems in the {\em Massively Parallel Computation} (MPC) model: (i) {\em grid graph connectivity}, (ii) approximate {\em Euclidean Minimum Spanning Tree} (EMST), and (iii) approximate DBSCAN. 

Our first result is a $O(1)$-round {\em Las Vegas} (i.e., succeeding with high probability) MPC algorithm for computing the {\em connected components} on a $d$-dimensional $c$-penetration grid graph ($(d,c)$-grid graph), where both $d$ and $c$ are positive integer constants. 
In such a grid graph, each vertex is a point with integer coordinates in $\mathbb{N}^d$, and an edge can only exist between two distinct vertices with $\ell_\infty$-norm at most $c$. 
To our knowledge, the current best existing result for computing the connected components (CC's) on $(d,c)$-grid graphs in the MPC model is to run the state-of-the-art MPC CC algorithms that are designed for general graphs: they achieve  
$O(\log \log n + \log D)$~\cite{DBLP:conf/focs/BehnezhadDELM19} and $O(\log \log n + \log \frac{1}{\lambda})$~\cite{DBLP:conf/podc/AssadiSW19} rounds, respectively, where $D$ is the {\em diameter} and $\lambda$ is the {\em spectral gap} of the graph.
With our grid graph connectivity technique, our second main result is a $O(1)$-round Las Vegas MPC algorithm for computing  {\em approximate Euclidean MST}. 
The existing state-of-the-art result on this problem is the $O(1)$-round MPC algorithm proposed by Andoni et al.~\cite{DBLP:conf/stoc/AndoniNOY14}, which only guarantees an approximation on the overall weight {\em in expectation}.  
In contrast, our algorithm not only guarantees a {\em deterministic overall weight approximation}, but also achieves
 a {\em deterministic edge-wise weight approximation}.
The latter property is crucial to many applications, such as finding the Bichromatic Closest Pair and Single-Linkage Clustering.
Last but not the least, our third main result is a $O(1)$-round Las Vegas MPC algorithm for computing an {\em approximate DBSCAN} clustering in $O(1)$-dimensional Euclidean space.

Given the importance of the aforementioned three problems, our proposed $O(1)$-round MPC algorithms may be of independent interest for designing new MPC algorithms for solving other related problems.
\end{abstract}
\begin{toappendix}

\section{Atomic MPC operations}
\label{app:op}

\subparagraph{Sorting.} In the MPC model, sorting can be $O(1)$ rounds.

\begin{fact}[\cite{DBLP:conf/isaac/GoodrichSZ11}]
\label{theo:sorting}
Sorting $n$ elements can be solved in {$O(\log_s n)$} round in the MPC model with $O(n)$ total memory and {$\Theta(s)\subseteq \Theta(n^{\alpha})$} local memory with any constant $\alpha\in(0,1)$. Specifically, the input is a set of $n$ comparable items. At the beginning, each machine stores $O(s)$ of the input data. There exists an algorithm that can run in {$O(\log_s n)$} round and leave the $n$ items sorted on the machines, i.e. the machine with smaller machine index holds a smaller part of $O(s)$ items. 
\end{fact}

\subparagraph{{Broadcasting.}} Consider a set of data $D$ with size of $|D|\leq s$ stored in one machine $\mathcal{M}$. 
There exists an algorithm in the MPC model that {\em broadcasts} the set of data $D$ to all machines. 
Specifically, when the algorithm terminates, each machine holds a copy of $D$. 
The broadcasting task can be done in two rounds if $\alpha \geq 1/2$, in which case $s \geq m$. 
In the first round, the machine $\mathcal{M}$ partitions $D$ evenly into $m$ parts, $D_1,\ldots, D_m$, 
and sends the $i$-th part to the $i$-th machine. 
In the second round, for $i$ in $[1,m]$, 
the $i$-th machine simultaneously sends $D_i$ to all the other machines. 
It can be verified that the amount of data sent and received in each round per machine is bounded by $|D|\in O(s)$.

\subparagraph{{Duplicate removal.}} 
Given $n$ items; some of them may be duplicates.
In the MPC model, removing all those duplicated items can be achieved in 
$O(\log_s n)$ rounds by two steps: (i) sort the $n$ items, (ii) delete an item if it equals to its {\em predecessor}. 

\subparagraph{{Sampling.}} 
Consider a set of data $V_1$ of size $n$ stored in a contiguous set of machines. 
Let $k=\beta s$ to be the desired size of the random subset of $V_1$, where  
$\beta$ is some universal constant.
Then a random subset of $V_1$ of size in the range of $[k,3k]$ can be obtained in $O(\log_s n)$ rounds with high probability at least $1 - \frac{1}{n^{\Omega(1)}}$.

Specifically, the sampling process works as follows.
In each machine, flip a coin for each element in $V_1$ stored in its local memory, where with probability $\frac{2k}{|V_1|}$ mark the element as a {\em candidate}, and record the candidate count.
Each machine sends the candidate count in it to machine $M_0$.
If the total candidate count is outside the range $[k, 3k]$, i.e., less than $k$ or greater than $3k$,
then start the whole process over again from scratch. 
Otherwise, $M_0$ instructs all the machines to to send their candidates to it. 
The resulting set of candidate stored in $M_0$ is a random subset of $V_1$.

According to the Chernoff bound~\cite{DBLP:books/daglib/0012859}, the success probability of obtaining a random subset of size in $[k, 3k]$ in one attempt is at least 
$1-2e^{\frac{-k}{8}}$, which achieves {\em high probability}.

\end{toappendix}

\section{Introduction}

Effective parallel systems for large scale datasets, such as MapReduce~\cite{DBLP:journals/cacm/DeanG08}, Dryad~\cite{DBLP:conf/eurosys/IsardBYBF07}, Spark~\cite{DBLP:conf/hotcloud/ZahariaCFSS10}, Hadoop~\cite{DBLP:books/daglib/0029284}, 
have received considerable attention in recent years. 
The {\em Massively Parallel Computation} (MPC) model~\cite{DBLP:conf/soda/KarloffSV10,DBLP:conf/isaac/GoodrichSZ11,DBLP:conf/pods/BeameKS13} has been proposed to provide a solid theoretical abstraction for the modern study of parallelism.

In this paper, we propose several $O(1)$-round Las Vegas algorithms, each succeeding with high probability, 
in the MPC model, for solving three fundamental problems:
(i) {\em connectivity} and {\em minimum spanning forest} (MSF) on multi-dimensional grid graphs,
(ii) {\em approximate Euclidean Minimum Spanning Tree} (EMST),
and 
(iii) {\em approximate DBSCAN}.
\subparagraph{The MPC Model.}
In the MPC model, the input is of size $n$ and {\em evenly} distributed among $m$ machines. 
Each machine is equipped with a {\em strictly sub-linear} local memory of size $\Theta(s)$, where $s=n/m = n^\alpha$ for some {\em constant} $\alpha \in (0, 1)$ 
~\cite{DBLP:conf/stoc/AndoniNOY14,DBLP:conf/soda/KarloffSV10,DBLP:conf/isaac/GoodrichSZ11}.
An MPC algorithm runs in {\em (synchronous) rounds};
each round proceeds two phases {\em one after another} : the {\em communication phase} first,  and then the {\em local computation phase}.
In the communication phase, each machine can send to and receive from other machines, in total,  $O(s)$ data. 
In the computation phase, each machine performs computation on the $O(s)$ data stored in its local memory,
and prepares what data can be sent to which machine
in the communication phase of the next round. 
The efficiency of an MPC algorithm is measured by the {\em total number of rounds}.

\paragraph*{Main Result 1: Connectivity on Grid Graphs}

We  consider a class of graphs called {\em $d$-dimensional $c$-penetration grid graphs} (for short, $(d,c)$-grid graphs).
Specifically, a graph $G= (V, E)$ is  a $(d,c)$-grid graph satisfies:
(i) the dimensionality $d \geq 2$ is a {\em constant} and $c \geq 1$ is an {\em integer};
(ii) each node in $V$ is a $d$-dimensional point with {\em integer coordinates} in $\mathbb{N}^d$ space;
(iii) for each edge $(u, v) \in E$, $0 < \|u, v\|_\infty \leq c$ holds, 
i.e., $u$ and $v$ are distinct vertices and their coordinates differ by at most $c$ on every dimension.
It can be verified that, when $c$ is a constant, by definition, each node in a $(d, c)$-grid graph can have {\em at most} $(2c + 1)^d - 1 \in O(1)$ neighbours, and hence, a $(d, c)$-grid graph is {\em sparse}, i.e., $|E| \in O(|V|)$. 
In particular, for $c = 1$, a $(d,1)$-grid graph is a common well-defined $d$-dimensional grid graph.

In real-world applications, grid graphs are often useful 
in the modeling of 3D meshes and terrains~\cite{burger2018fast}.
More importantly, grid graphs are  
extremely useful in many {\em algorithmic designs} for solving various computational geometry problems, such as Approximate Nearest Neighbour Search~\cite{arya1998optimal}, Euclidean Minimum Spanning Tree~\cite{indyk2004algorithms}, DBSCAN~\cite{agrawal1998automatic} and etc.
Therefore, the study of algorithms on grid graphs has become an important topic in the research community.
Our first result is this theorem:
\vspace{-1mm}
\begin{theoremapprep}\label{thm:grid-cc}
	Given a $(d, c)$-grid graph $G = (V, E)$ with $c \in O(1)$, 
	there exists a Las Vegas MPC algorithm with local memory {$\Theta(s) \subseteq \Theta(|V|^\alpha)$} per machine, for an arbitrary constant $\alpha \in (\max\{ \frac{d + 1}{d + 2}, \frac{4}{5}\}, 1)$, which computes all the {\em connected components} of $G$ in 
$O(1)$
rounds {\em with high probability}. Specifically, the algorithm assigns the ID of the connected component to each node, of which the node belongs to. 
\end{theoremapprep}
\begin{proof}
This Theorem follows immediately from Theorem~\ref{thm:implicit-connectivity} in Section~\ref{sec:grid-cc}.
\end{proof}

Computing connected components (CC's) of sparse graphs in the MPC model is notoriously challenging.
In particular, the well accepted {\em 2-CYCLE Conjecture}~\cite{yaroslavtsev2018massively,nanongkai2022equivalence} says that:
\begin{quote}
Every MPC algorithm, using $n^{1 - \Omega(1)}$ local memory per machine and $O(n)$ space in total, 
requires $\Omega(\log n)$ rounds to correctly solve, with high probability, the {\em 1-vs-2-Cycles problem} which asks to distinguish whether the input graph consists of just one cycle with $n$ nodes or 
two cycles with $n/2$ nodes each.
\end{quote}
Theoretically speaking, the 2-CYCLE Conjecture can be disproved as ``easy'' as finding just a single constant $\alpha < 1$ 
and an algorithm that solves the 1-vs-2-Cycles problem with $O(n^\alpha)$ per-machine local memory and $O(n)$ total space in $o(\log n)$ rounds.
However, according to the recent hardness results, the conjecture is robust and hard to be refuted~\cite{nanongkai2022equivalence}.

Even worse, since the input instance of the 1-vs-2-Cycles problem is so simple, 
it eliminates the hope of computing CCs for many classes of ``simple'' graphs, 
such as path graphs, tree graphs and planar graphs, 
which are often found to be simpler cases for many fundamental problems. 
As a result, it is still very challenging to find a large class of graphs on which the CC problem can be solved in $o(\log n)$ rounds. 

The existing state-of-the-art works~\cite{DBLP:conf/focs/AndoniSSWZ18, DBLP:conf/podc/AssadiSW19,DBLP:conf/focs/BehnezhadDELM19,
DBLP:conf/stoc/CoyC22} 
consider general sparse graphs that are parameterised by the {\em diameter} $D$ and the {\em spectral gap} $\lambda$.
Specifically, Andoni et al.~\cite{DBLP:conf/focs/AndoniSSWZ18} propose a~$O(\log D\log\log n)$-round randomized algorithm, where
$n$ is the number of nodes. 
Assadi et al.~\cite{DBLP:conf/podc/AssadiSW19} give an algorithm with $O(\log \log n+\log \frac{1}{\lambda})$ rounds.
Recently, Behnezhad et al.~\cite{DBLP:conf/focs/BehnezhadDELM19} improve the bound by Andoni et al.\ to $O(\log D + \log\log n)$ randomized, 
while Coy and Czumaj~\cite{DBLP:conf/stoc/CoyC22} propose a deterministic MPC algorithm achieving the same round number bound.
However, all these state-of-the-art general algorithms still require $O(\log n)$ rounds for solving the CC problem on $(d,c)$-grid graphs 
in the worst case, as the diameter $D$ of a $(d,c)$-grid graph can be as large as $n$
and the spectral gap, 
$\lambda$, can be as small as $\frac{1}{n}$.

Therefore, our Theorem~\ref{thm:grid-cc} is not only an immediate improvement over these state-of-the-art known results the CC problem on $(d,c)$-grid graphs, 
but, interestingly, also suggests a large class of sparse graphs, $(d,c)$-grid graphs with $c \in O(1)$,
on which the CC problem can be solved 
in $O(1)$ rounds.

We note that it is the {\em geometric property} of $(d,c)$-grid graphs making them admit $O(1)$-round CC algorithms.
As we shall discuss in Section~\ref{sec:grid-cc},
our key technique is the so-called {\em Orthogonal Vertex Separator} for multi-dimensional grid graphs proposed by Gan and Tao~\cite{DBLP:journals/jgaa/GanT18}
which was originally proposed to efficiently compute CC's on $(d, 1)$-grid graphs in {\em External Memory} model~\cite{DBLP:journals/cacm/AggarwalV88}.
However, in the MPC model, it appears difficult to compute such separators in $O(1)$ rounds.
To overcome this, we relax the Orthogonal Vertex Separator to 
a ``{\em weaker yet good-enough}''
version for our purpose of designing $O(1)$-round MPC algorithms and  
extend this technique~\footnote{
A similar technique of using small-size geometric separators (though they are not the same as ours) is applied to computing Contour Trees on terrains~\cite{nath2016massively} in $O(1)$ rounds in the MPC model.
Moreover, the idea of a crucial technique in our separator computation, 
$\varepsilon$-approximation, 
actually stems from the KD-tree construction in the MPC model in~\cite{DBLP:conf/pods/AgarwalFMN16}.   
}
to $(d,c)$-grid graphs which are crucial to our $O(1)$-round EMST and DBSCAN algorithms.

Moreover, with standard modifications, our CC algorithm can also compute the {\em Minimum Spanning Forests} (MSF)  on $(d,c)$-grid graphs:
\begin{corollaryapprep}\label{cor:grid-msf}
	Given an {\em edge-weighted} $(d,c)$-grid graph $G = (V, E)$ with $c \in O(1)$,
	there exists a Las Vegas MPC algorithm with local memory {$\Theta(s) \subseteq \Theta(|V|^\alpha)$} per machine, for arbitrary constant $\alpha \in (\max\{\frac{d+ 1}{d + 2}, \frac{4}{5}\}, 1)$, 
	which computes a {\em Minimum Spanning Forest} in 
$O(1)$ 
rounds with high probability.
\end{corollaryapprep}
\begin{proof}
This Corollary follows immediately from Theorem~\ref{thm:implicit-connectivity} in Section~\ref{sec:grid-cc}.
\end{proof}

\paragraph*{Main Result 2: Approximate EMST}

\noindent
By enhancing our MSF technique for $(d,c)$-grid graphs, 
our second main result is a new $O(1)$-round Las Vegas MPC algorithm for computing approximate {\em Euclidean Minimum Spanning Tree} (EMST).
Specifically, the problem is defined as follows.
Given a set, $P$, of $n$ points with integer coordinates in $d$-dimensional space $\mathbb{N}^d$, where $d$ is a constant and the coordinate values are in $[0, \Delta]$ with $\Delta = n^{O(1)}$,
let $G = (V, E)$ be an {\em implicit} edge-weighted complete graph on $P$ such that: $V = P$ and 
the weight of each edge $(u,v)\in E$ is the Euclidean distance between $u$ and $v$.
The goal of the EMST problem is to compute an MST on the implicit complete graph, $G$, of $P$.
The EMST problem is one of the most fundamental and important problems in computational geometry and 
has sparked significant attention for machine learning~\cite{DBLP:journals/corr/abs-1810-04805,DBLP:conf/cvpr/HeZRS16,DBLP:journals/corr/abs-1301-3781}, data clustering~\cite{DBLP:journals/corr/abs-2307-07848,DBLP:conf/spaa/CoyCM23} and optimization problems~\cite{DBLP:conf/soda/CzumajLMS13,DBLP:conf/approx/Monemizadeh23}. 
The existing state-of-the-art result is an $O(1)$-round {\em Monte-Carlo} MPC algorithm proposed by Andoni et al.~\cite{DBLP:conf/stoc/AndoniNOY14} for computing an approximate EMST. 
Their algorithm possesses the following properties:
(i) it can achieve $O(1)$ rounds in the worst case and works for all constant $\alpha \in (0, 1)$; and 
(ii) {\em in expectation}, it returns a spanning tree $T$ of the implicit graph $G$ whose total edge weight is at most $(1 + \rho)$ times of the total edge weight of the exact EMST $T^*$, where the approximation factor $\rho > 0$ is a {\em constant}.
In comparison, we have the following theorem:
\begin{theoremapprep}\label{thm:emst}
	Given a set $P$ of $n$ points in $d$-dimensional space $\mathbb{N}^d$ with coordinate values in $[0, \Delta]$ for $\Delta = n^{O(1)}$ and a constant approximation factor $\rho > 0$, 
	there exists a Las Vegas MPC algorithm with local memory $\Theta(s) \subseteq \Theta(n^\alpha)$ per machine, for arbitrary constant $\alpha \in (\max\{\frac{d + 1}{d+2}, \frac{4}{5}\}, 1)$, which computes a spanning tree $T$ of the implicit graph $G$ of $P$ in 
 $O(1)$
 rounds with high probability 
	and $T$ {\em deterministically} satisfyies the following: 
	\begin{itemize}
		\item {\bf Overall Weight Approximation}: the total edge weight of $T$ is at most $(1 + \rho)$ times the total edge weight of the exact EMST $T^*$;
		\item {\bf Edge-Wise Weight Approximation}: for every edge $(u,v)$ in $T^*$ with weight $w(u,v)$, 
there must exist a path from $u$ to $v$ in $T$ such that the \underline{maximum weight} of the edges on this path is at most $(1 + \rho) \cdot w(u,v)$. 
\end{itemize}
\end{theoremapprep}
\begin{proof}
This Theorem follows from all the analysis and lemmas in Section~\ref{sec:emst}.
\end{proof}

Our algorithm improves Andoni et al.'s approximate EMST algorithm in two folds.
First, our algorithm computes a feasible $\rho$-approximate EMST deterministically, while theirs can only guarantee this in expectation.
Second, and most importantly, the approximate EMST computed by our algorithm can further guarantee the edge-wise approximation which Andoni et al.'s algorithm cannot achieve. 
We note that the edge-wise approximation guarantee is crucial to solving many down-stream algorithmic problems, such as {\em $\rho$-approximate bichromatic closest pair}~\cite{agarwal1990euclidean}, Single-Linkage Clustering~\cite{zhou2011clustering} and other related problems~\cite{xue2019colored,ding2011solving}.

\paragraph*{Main Result 3: Approximate DBSCAN}

\noindent
Our third
main result is a $O(1)$-round MPC algorithm for computing $\rho$-approximate DBSCAN clustering as defined in~\cite{gan2015dbscan}.
To our best knowledge, this is the first $O(1)$-round algorithm in the MPC model for solving the DBSCAN problem.

\begin{theoremapprep}\label{thm:dbscan}
	Given a set $P$ of $n$ points in $d$-dimensional space $\mathbb{R}^d$, where $d$ is a constant,
	two DBSCAN parameters: $\eps$ and $\minPts$, and a constant approximation factor $\rho > 0$,
	there exists a Las Vegas MPC algorithm with local memory {$\Theta(s) \subseteq \Theta(n^\alpha)$} per machine, for arbitrary constant $\alpha \in (\max\{\frac{d + 1}{d+2}, \frac{4}{5}\}, 1)$, which computes 
	an $\rho$-approximate DBSCAN clustering of $P$ in 
 {$O(1)$} rounds with high probability.
\end{theoremapprep}
\begin{proof}
This Theorem follows immediately from all the analysis and lemmas in Section~\ref{sec:ourMPCalgos}.
\end{proof}

\subparagraph{Derandomization}
As we shall see shortly, the randomness of all our Las Vegas algorithms only comes from the {\em randomized computation} of an {\em $\frac{1}{r}$-approximation} for some $r$.
Indeed, it is known that such approximations can be computed in $O(1)$ MPC rounds deterministically~\cite{DBLP:conf/pods/AgarwalFMN16}.
And therefore, all our algorithms can be {\em derandomized}, yet slightly shrinking the feasible value range of the local memory parameter, $\alpha$, to $(\max\{\frac{d + 1}{d+2}, \frac{6}{7}\}, 1)$.

In this paper, 
the proofs of all the theorems, lemmas and claims marked with \emph{\mainbodyrepeatedtheorem} can be found in Appendix.

\section{Grid Graph Connectivity and MSF}
\label{sec:grid-cc}

In this section, we introduce our $O(1)$-round MPC algorithms for computing CC's and MSF's on $(d,c)$-grid graphs with $c \in O(1)$ with high probability.

\subparagraph{Implicit Graphs and Implicit $(d,c)$-Grid Graphs.} 
For the ease of explanation of our techniques for solving the EMST problem in the next section, 
we first introduce the notion of \textit{implicit graphs}, denoted by $G=(V,\mathcal{E})$.
\vspace{-1mm}
\begin{definition}[Implicit Graph]
	An {\em implicit graph} $G = (V, \mathcal{E})$ is a graph consisting of a {\em node set} $V$ and an {\em edge formation rule}, $\mathcal{E}$, where:
\begin{itemize}
	\item each node $u$ in $V$ is associated with $O(1)$ words of information,
 e.g., $O(1)$-dimensional coordinates;
	\item the edge formation rule, $\mathcal{E}$, is a function (of size bounded by $O(1)$ words) on distinct node pairs:
		for any two distinct nodes  $u, v \in V$, $\mathcal{E}(u,v)$ returns, based on the $O(1)$-word information associated with $u$ and $v$: (i) whether edge $(u,v)$ exists in $G$, and (ii) the weight of $(u,v)$ if edge $(u,v)$ exists and the weight is well defined. 
\end{itemize}
\end{definition}

By the above definition, the size of every implicit graph is bounded by $O(|V|)$.
An implicit graph can be converted to an {\em explicit graph} by materialising 
all the edges with the edge formation rule $\mathcal{E}$.
If the corresponding explicit graph of an implicit graph $G=(V,\mathcal{E})$ is a $(d,c)$-grid graph, 
we say $G=(V, \mathcal{E})$  is an {\em implicit $(d,c)$-grid graph}.
For example, given a set $P$ of $n$ points in $\mathbb{N}^d$ with coordinate values in $[0, \Delta]$, 
the {\em implicit complete Euclidean graph} $G_{\text{comp}}$ of $P$ can be considered as an implicit $(d, \Delta)$-grid graph $G = (P, \mathcal{E})$, 
where $\mathcal{E}(u,v)$ returns $dist(u,v)$, the Euclidean distance between $u$ and $v$, as the edge weight.
Another example is the unit-disk graph~\cite{clark1990unit} in $\mathbb{N}^d$ space with $l_\infty$-norm, which can be essentially regarded as the {\em implicit $(d,1)$-grid graph} with an edge formation rule: 
there exists an edge between two vertices if and only if their $l_\infty$-norm distance is at most~$1$.
However, as we will see in our EMST algorithm, the notion of implicit $(d,c)$-grid graphs is more general 
due to the flexibility on the edge formation rule specification.
For example, 
in addition to the coordinates, 
each node can also bring an ID of the connected component that it belongs to, 
and the edge formation rule can further impose a rule that there must be no edge between two nodes having the same connected component ID.

\subsection{Pseudo $s$-Separator and Its Existence}

\begin{figure}[t]
  \begin{center}
  \includegraphics[scale=0.16]{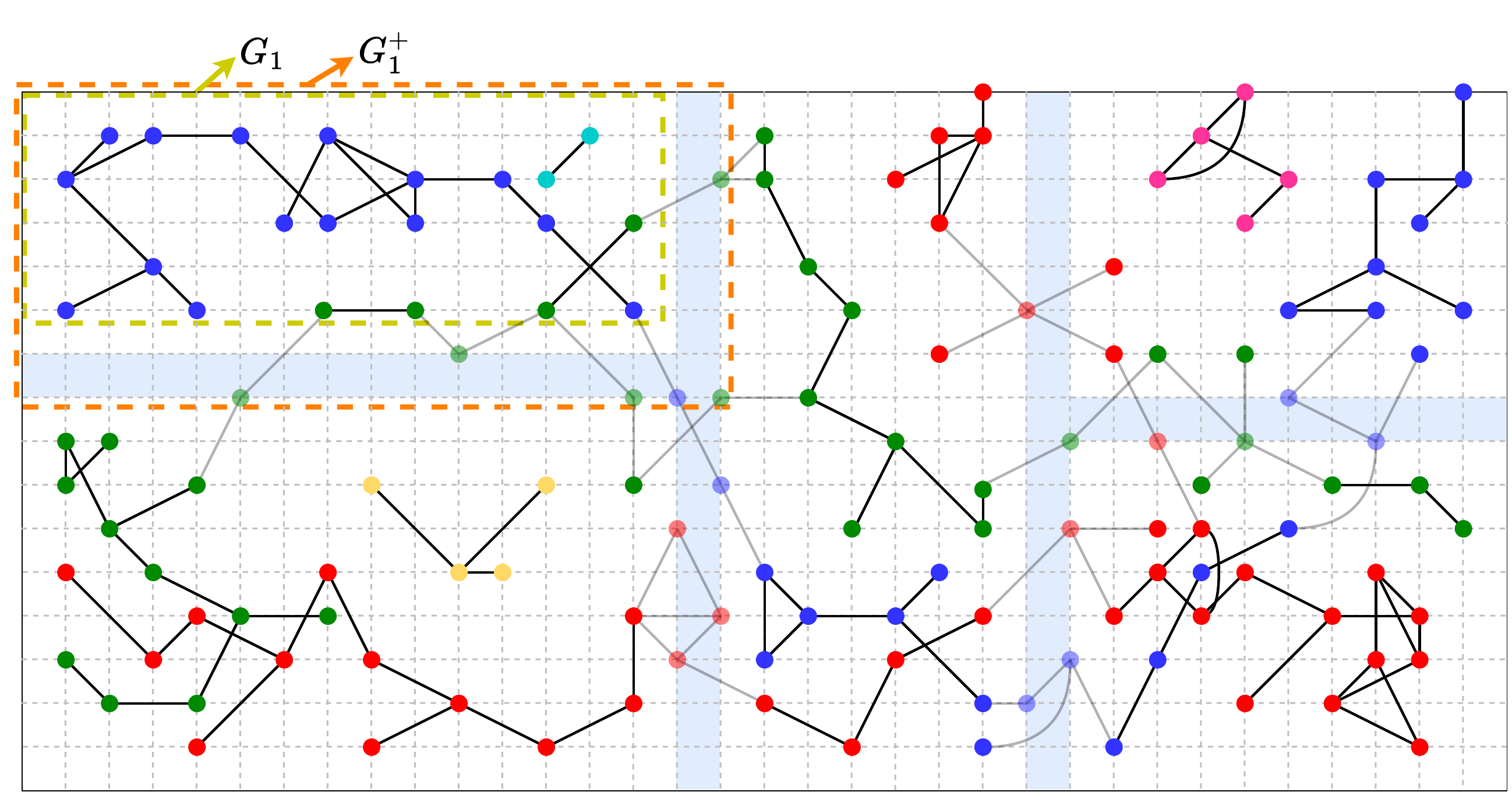}
  \end{center}
  \vspace{-3mm}
  \caption{A $(d,c)$-grid graph $G$ with $d =2$ and $c=2$ and a separator.
 The blue regions are the target $c$-dividers: the vertices in them are   
	  separator vertices.
	  Removing all the separator vertices disconnects $G$ into $5$ sub-graphs.
  }
  \label{fig:grid-sep}
  \vspace{-5mm}
\end{figure}

In this subsection, we extend the notion of {Orthogonal Vertex Separator}~\cite{DBLP:journals/jgaa/GanT18} to a {\em weaker yet good-enough} version for our algorithm design purpose for solving the CC and MSF problems on implicit $(d,c)$-grid graphs in the MPC model. 
We call this new notion a {\em Pseudo $s$-Separator}; recall that $O(s)$ is the size of local memory, per machine, in the MPC model:

\begin{definition}[Pseudo $s$-Separator]\label{def:s-sep}
	Consider an implicit $(d,c)$-grid graph $G=(V, \mathcal{E})$ and a parameter  $s<|V|$.
	A set $S \subseteq V$ is a {\em pseudo $s$-separator} of $G$ if it satisfies:  
	\begin{itemize}
	\item $|S| \in O(\frac{c \cdot |V|}{s^{1/d}} \cdot \log s)$
	\item Removing the vertices in $S$ disconnects $G$ 
		into $h\in O(|V|/s)$ sub-graphs $G_i=(V_i,\mathcal{E})$, $i \in [1,h]$, such that $V_i \cap V_j = \emptyset$ for all $i\neq j$, $V_i \subset V$ and $|V_i| \in O(s)$.
\end{itemize}
\end{definition}

\bfrev{
\begin{definition}[Pseudo $s$-Separator]\label{def:s-sep}
	Consider an implicit $(d,c)$-grid graph $G=(V, \mathcal{E})$ with $1 \leq c \leq s^{\frac{1}{d^3}}$ and $|V| > s$.
	A set $S \subseteq V$ is a {\em pseudo $s$-separator} of $G$ if it satisfies:  
	\begin{itemize}
	\item $|S| \in O(\frac{c \cdot |V|}{s^{1/d}} \cdot \log s)$
	\item Removing the vertices in $S$ disconnects $G$ 
		into $h\in O(|V|/s)$ sub-graphs $G_i=(V_i,\mathcal{E})$, $i \in [1,h]$, such that $V_i \cap V_j = \emptyset$ for all $i\neq j$, $V_i \subset V$ and $|V_i| \in O(s)$.
\end{itemize}
\end{definition}
}
Figure~\ref{fig:grid-sep} shows an example. Next, we show that if $1 \leq c \leq s^{\frac{1}{d^3}}$ and 
$|V| >  s$, a pseudo $s$-separator must 
exist in an implicit $(d,c)$-grid graph $G=(V, \mathcal{E})$. 
While our proof mimics that in~\cite{DBLP:journals/jgaa/GanT18}, it does require some new non-trivial ideas. 
We introduce some useful notations.

\subparagraph{Minimum Bounding Space of $G$.}
Consider an implicit $(d,c)$-grid graph $G= (V, \mathcal{E})$;
let $[x^{(i)}_{\min}, x^{(i)}_{\max}]$ be the value range of $V$ on the $i$-th dimension (for $i \in \{1, 2, \ldots, d\}$), 
where $x^{(i)}_{\min}$ and $x^{(i)}_{\max}$ are the {\em smallest} and {\em largest} coordinate values of the points in $V$ on dimension~$i$, respectively. The {\em minimum bounding space} of $G$ 
is defined as 
$\text{mbr}(V) = [x^{(1)}_{\min}, x^{(1)}_{\max}] \times [x^{(2)}_{\min}, x^{(2)}_{\max}] \times  \cdots \times [x^{(d)}_{\min} , x^{(d)}_{\max}]$.
Moreover, we denote the projection of $\text{mbr}(V)$ on the $i$-th dimension by $\text{mbr}(V)[i] = [x^{(i)}_{\min}, x^{(i)}_{\max}]$.

\subparagraph{$c$-Divider.}
Consider the minimum bounding space $\text{mbr}(V)$;
a {\em $c$-divider} on the $i$-th dimension at coordinate $x$, denoted by $\pi(i, x)$, is an
axis-parallel rectangular range whose projection on dimension $i$ is $[x, x+c -1]$ and the projection on dimension $j \neq i$ is $\text{mbr}(V)[j]$.
With a $c$-divider $\pi(i,x)$, the vertex set $V$ is partitioned into three parts:
(i) the {\em left part} $ V_{\text{left}} = \{ u\in V \mid u[i] \leq x -1\}$, 
(ii) the {\em right part} $V_{\text{right}} = \{ u \in V \mid u[i] \geq x + c\}$, 
and
(iii) the {\em boundary part} $S_{\pi} = V \cap \pi(i,x)$,
where $u[i]$ is the $i^\text{th}$ coordinate of a vertex $u$.

We say that $\pi(i,x)$ {\em separates} $G$ into two sub-graphs induced by the left and right part, i.e., $V_{\text{left}}$ and $V_{\text{right}}$, respectively.
And all the vertices in $S_{\pi}$ are called {\em separator vertices}.
Observe that for any $u \in V_{\text{left}}$ and any $v \in V_{\text{right}}$, 
$|u[i] - v[i]| \geq c + 1$ implying 
$\|u,v\|_\infty \geq c + 1$.
As a result,
there must not exist any edge between such~$u$ and~$v$.
With such notation, we show the following Binary Partition Lemma.

\begin{lemma}[Binary Partition Lemma]\label{lmm:bin-partition}
	Consider an implicit $(d,c)$-grid graph $G = (V, \mathcal{E})$ with 
	$1 \leq c \leq s^{\frac{1}{d^3}}$ and $|V| > s$; 
	there exists a $c$-divider $\pi(i,x)$ partitioning $V$ into $\{S_{\pi}, \{V_{\text{left}}, V_{\text{right}}\}\}$ such that:
		\begin{itemize}
			\item $|V_{\text{left}}| \geq \frac{|V|}{4(d+1)}$ and $|V_{\text{right}}| \geq \frac{|V|}{4(d+1)}$, and 
			\item $|S_{\pi}| \leq 2c(1 + d) ^\frac{1}{d} |V|^{1 - \frac{1}{d}}$. 
		\end{itemize}
\end{lemma}

\begin{proof}
We prove this lemma with a constructive algorithm.
Given the vertex set, $V$, our algorithm first finds two integers, $y_j$ and $z_j$, for each dimension $j\in [1,d]$. 
Specifically, $y_j$ is the {\em largest} integer that the number of the vertices in $V$ located to the ``left'' of $y_j$ is at most $\frac{|V|}{2(1+d)}$, 
while $z_j$ is the {\em smallest} integer that the number of the vertices in $V$ located to the ``right'' of $z_j$ is at most $\frac{|V|}{2(1+d)}$. 
Formally, $y_j=\max\{x\in\N \,:\, |\{v\in V \mid v[j]<x\}|\leq \frac{|V|}{2(1+d)}\}$ and $z_j=\min\{x\in\N \,:\, |\{v\in V \mid v[j]>x\}|\leq \frac{|V|}{2(1+d)}\}$,
where $v[j]$ is the coordinate of $v$ on dimension $j$.

Consider the hyper-box in $\N^d$, whose projection on each dimension is $[y_j,z_j]$ for $j \in [1, d]$. 
By the definition of $y_j$ and $z_j$, 
this box 
contains in total at least $|V|(1-\frac{2d}{2(d+1)}) = \frac{|V|}{d+1}$ vertices in $V$. 
Since the coordinates of the points are all integers, 
the volume of the box must be at least $\frac{|V|}{d+1}$. That is, $\prod_{j=1}^{d}(z_j-y_j+1)\geq \frac{|V|}{d+1}$. 
Thus, there must exist some dimension $i$ such that 
$z_i-y_i+1 \geq (\frac{|V|}{1+d})^{\frac{1}{d}}$. 
Next we fix this dimension,~$i$.
Let $S_{\pi(i,x)}$ be the set of all vertices in $V$ falling in a $c$-divider $\pi(i,x)$.
Within dimension $i$, each vertex in $V$ can be contained in at most $c$ {\em consecutive} $c$-dividers.
As a result, we have $\sum_{x'\in[y_i,z_i-c]} |S_{\pi(i,x')}| \leq c|V|$. 
Since there are $z_i-c-y_i+1$ terms in the summation, 
 there exists an integer $x \in [y_i,z_i-c]$ with 
	\begin{equation*}
|S_{\pi(i,x)}| \leq \frac{c|V|}{z_i-y_i+1 -c} \leq \frac{c|V|}{(\frac{|V|}{1+d})^{\frac{1}{d}}-c}\leq2c(1+d)^{\frac{1}{d}}|V|^{1-\frac{1}{d}}\,,
\end{equation*}
\noindent
where
the second inequality comes from the bound on~$z_i-y_i+1$ and the third inequality is due to $c \leq s^{\frac{1}{d^3}} \leq \frac{1}{2}\cdot (\frac{|V|}{(1 + d)})^{\frac{1}{d}}$ as $|V| > s$.
Then such $\pi(i, x)$ is a desired $c$-divider. 

Next, we bound the size of $V_{\text{left}}$ with respect to $\pi(i, x)$: bounding $|V_{\text{right}}|$ is analogous. 
Since $x \geq y_i$, there are only two possible relations between $x$ and $y_i$:
\begin{itemize}
    \item If $x >  y_i$, according to the definition of $y_i$, the number of vertices in $V$ to the left of $\pi$ is at least $|V|(\frac{1}{2d+2})\geq \frac{|V|}{4d+4}$.
    \item If $x=y_i$, the number of vertices in $V$ to the left of $\pi$ is at least $\frac{|V|}{2d+2}-|S_{\pi(i,x)}|\geq |V|\left(\frac{1}{2d+2}-\frac{2c(d+1)^{\frac{1}{d}}}{|V|^{\frac{1}{d}}}\right)\geq \frac{|V|}{4d+4}$.
\end{itemize}
Either way, $|V_{\text{left}}| \geq \frac{|V|}{4(d+1)}$ holds;
Lemma~\ref{lmm:bin-partition} thus follows.
\end{proof}

\subparagraph
{The Multi-Partitioning Algorithm.} Given an implicit $(d,c)$-grid graph $G = (V, \mathcal{E})$ with $1 \leq c \leq s^{\frac{1}{d^3}}$ and $|V| > s$, initialize $S \leftarrow \emptyset$ and perform the following steps:
\begin{itemize}
\item apply Lemma~\ref{lmm:bin-partition} on $G$ to obtain $\{S_{\pi}, \{V_{\text{left}}, V_{\text{right}}\}\}$;
		\item $S \leftarrow S \cup S_{\pi}$;
		\item if $|V_{\text{left}}| > s$, recursively apply Lemma~\ref{lmm:bin-partition} on $G_{\text{left}} = (V_{\text{left}}, \mathcal{E})$;
		\item if $|V_{\text{right}}| > s$, recursively apply Lemma~\ref{lmm:bin-partition} on $G_{\text{right}} = (V_{\text{right}}, \mathcal{E})$;
\end{itemize}

\begin{lemmaapprep}\label{lmm:s-sep-ext}
	The vertex set $S$, obtained by the above \emph{Multi-Partitioning Algorithm}, is a pseudo $s$-separator of the implicit $(d,c)$-grid graph $G = (V, \mathcal{E})$, where $1 \leq c \leq s^{\frac{1}{d^3}}$ and $|V| > s$. 
\end{lemmaapprep}
\begin{proof}
Let $V_1, V_2, \ldots, V_h$ be the vertex sets of the resulting sub-graphs from the multi-partitioning algorithm.
By definition, $V_1, V_2, \ldots V_h$ form a partition of $V \setminus S$ and, hence, are mutually disjoint.
Second, by the recursion condition and Lemma~\ref{lmm:bin-partition}, 
we know that $\frac{s}{4d + 4} \leq |V_i| \leq s$ holds for all $i \in [1, h]$, and therefore, $h \in O(\frac{|V|}{s})$, which is
 the second bullet of Definition~\ref{def:s-sep}.

Next, we bound the size of $S$. The algorithm essentially constructs a {\em conceptual} space-partitioning binary tree $\mathcal{T}$, 
where each internal node represents a $c$-divider,
and each leaf node represents an induced sub-graph $V_i$, for $i \in [1, h]$.
According to Lemma~\ref{lmm:bin-partition}, we know that the size of each such $c$-divider is at most 
$\frac{2c(1+d)^{\frac{1}{d}} |V'|}{|V'|^{\frac{1}{d}}}\leq 
\frac{2c(1+d)^{\frac{1}{d}} |V'|}{s^{\frac{1}{d}}}$, where $V'$ is vertex set of the sub-graph the $c$-divider separates.
The sum of all $|V'|$ terms at the same level is at most $|V|$, and the level number is at most $\log_{\frac{4d + 3}{4d+4}} {\frac{|V|}{s/(4d + 4)}} \in  O(\log s)$ because $s = |V|^{\alpha}$ for some constant $\alpha \in (0,1)$.
Therefore, 
$|S|$ is bounded by $O(\frac{c \cdot |V|}{s^{1/d}} \cdot \log s)$.
Lemma~\ref{lmm:s-sep-ext} thus follows.
\end{proof}
A construction algorithm proves existence, hence:
\begin{theorem}
For any implicit $(d,c)$-grid graph $G = (V, \mathcal{E})$ with $1 \leq c \leq s^{\frac{1}{d^3}}$ and $|V| > s$,
a pseudo $s$-separator $S$ of $G$ must exist.
\end{theorem}

\subparagraph{Remark.}
While our Pseudo $s$-Separator (PsSep) looks similar to the notation of the Orthogonal Vertex Separator (OVSep) proposed by Gan and Tao~\cite{DBLP:journals/jgaa/GanT18}, 
we emphasize that 
some crucial differences between them are worth noticing.
First, 
the OVSep was originally proposed for solving problems on $(d,1)$-grid graphs (in the External Memory model). Our PsSep supports implicit $(d, c)$-grid graphs for parameter $c$ up to a non-constant $s^{\frac{1}{d^3}}$ in the MPC model.
This extension, discussed in Section~\ref{sec:emst}, plays a significant role in our $O(1)$-round algorithm for computing Approximate EMST's. 
Second and most importantly,
it is still unclear how (and appears difficult) to compute the OVSep in $O(1)$ MPC rounds.
In contrast, our PsSep is proposed to overcome this technical challenge and can be computed in $O(1)$ MPC rounds.
However, as discussed in the next sub-section, 
because of the application of $\eps$-approximation in the construction algorithm of PsSep,
the separator size at each recursion level no longer geometrically decreases 
and therefore,  
leading to a logarithmic blow-up in the overall separator size in PsSep. 
Despite of this size blow-up, by setting the parameters carefully, we show that our PsSep can still fit in the local memory of one machine, 
and thus, it is still {\em good enough} for the purpose of computing CC's on implicit $(d,c)$-grid graphs and solving the Approximate EMST and Approximate DBSCAN problems in $O(1)$ rounds in the MPC model.

\subsection{$O(1)$-Round Pseudo $s$-Separator Algorithm}
\label{subsec:pseudoseparator}
In this subsection, we show a $O(1)$-round MPC algorithm for computing a pseudo $s$-separator for implicit $(d,c)$-grid graphs with $1\leq c\leq s^{\frac{1}{d^3}}$.
The Multi-Partitioning Algorithm proves the existence of a pseudo $s$-separator,
but a straightforward simulation of this algorithm in the MPC model is, however, insufficient to achieve $O(1)$ rounds. 
To overcome this, we 
resort to the technique of {\em $\eps$-approximation}~\cite{DBLP:journals/jcss/LiLS01,DBLP:conf/pods/AgarwalFMN16}.

\subsubsection{Preliminaries}
\subparagraph{Range Space and $\eps$-Approximation.}
A {\em range space} $\Sigma$ is a pair $(X, \mathcal{R})$, where $X$ is a ground set and $\mathcal{R}$ is a family of subsets of $X$. The elements of $X$ are {\em points} and the elements of $\mathcal{R}$ are {\em ranges}.
For $Y\subseteq X$, the {\em projection} of $\mathcal{R}$ on $Y$ is $\mathcal{R}_{|Y}=\{q \cap Y \mid  q \in \mathcal{R}\}$.
Given a range space $\Sigma=(X,\mathcal{R})$ and $\eps\in[0,1]$, a subset $X'\subseteq X$ is called an {\em $\eps$-approximation} of $\Sigma$ if for every range $q \in \mathcal{R}$, we have $\left|\frac{|X'\cap q|}{|X'|}-\frac{|q|}{|X|}\right|\leq \eps\,.$

\begin{fact}[\cite{DBLP:journals/jcss/LiLS01,DBLP:conf/pods/AgarwalFMN16}]
\label{fact}
	Consider the range space $(V,\mathcal{R})$, where the ground set $V$ is a set of $n$ points in $\mathbb{N}^d$ 
and $\mathcal{R}=\{V\cap \Box \mid \forall \text{ axis-parallel  rectangular range } \Box \text{ in } \N^d\}$. 
A random sample set $\tilde{V}$ of $V$ with size $|\tilde{V}| \in \Theta({r^2} \cdot \log n)$ is a $\frac{1}{r}$-approximate of $(V, \mathcal{R})$ with high probability.
\end{fact}

\noindent
Throughout this paper, we consider axis-parallel rectangular ranges only.
For simplicity, we may just say $\tilde{V}$ is a $\frac{1}{r}$-approximation of $V$. 

\subsubsection{Our Pseudo $s$-Separator Algorithm}
Our MPC algorithm comprises several {\em super rounds}, each performing:
\begin{itemize}
	\item Obtain an $\eps$-approximation, $\tilde{V}$, tuning $\eps$ so this fits in the local memory of a machine; 
	\item Perform roughly $O(s^{O(1)})$ ``good enough'' binary partitions in the local memory. 
\end{itemize}
Our algorithm then {\em simultaneously} recurs on  each of the resulting induced sub-graphs, if applicable, in the next super round.
As we will show, each super round would decrease the size of the graph by a factor of $\Theta(s^{O(1)})$, 
and there can be at most $O(\log_{s^{O(1)}} \frac{|V|}{s}) = O(1)$ super rounds.
Moreover, 
by Fact~\ref{fact}, an $\epsilon$-approximation can be obtained by sampling 
which can be achieved in $O(1)$ MPC rounds; detailed implementations can be found in Appendix~\ref{app:op}. 
Therefore, our algorithm runs in $O(1)$ MPC rounds.

Define $r = 2 s^{1/d}$.
In order to ensure that a $\frac{1}{r}$-approximation $\tilde{V}$ of $V$ fits the local memory size, $O(s)$, 
without loss of generality, in the following, we assume that the dimensionality, $d$, is at least~$3$.
Otherwise, for the case $d=2$, we can ``lift'' the dimensionality to~$3$ by adding a {\em dummy} dimension to all the points in $V$; however, the bounds derived for $d = 3$ apply.

The challenge lies in how to find a {\em good enough} (compared to that in Lemma~\ref{lmm:bin-partition}) $c$-divider $\pi$ 
to separate $V$ in the local memory with access to a $\frac{1}{r}$-approximation $\tilde{V}$ {\em only}.
\begin{lemmaapprep}\label{lmm:approx-bin-partition}
	Given $V$ with $|V| > s$ and $d \geq 3$, with only access to a $\frac{1}{r}$-approximation $\tilde{V}$ of $V$ such that $r = 2 s^{1/d}$ and $|\tilde{V}| \in O(s)$, we can find a $c$-divider $\pi$ in $\text{mbr}(V)$, which satisfies:
	\begin{itemize}
		\item $|V_{\text{left}}| \geq \frac{|V|}{8(d+1)}$ and  $|V_{\text{right}}| \geq \frac{|V|}{8(d+1)}$,
		\item $|S_{\pi}| \leq 4c(1 + d)^{1/d}|V|s^{-1/d}$.
	\end{itemize}
\end{lemmaapprep}
\begin{proof}
	By Lemma~\ref{lmm:bin-partition}, we know that there must exist a $c$-divider, $\pi^*$, achieving the two bullets in Lemma~\ref{lmm:bin-partition}. By the property of $\frac{1}{r}$-approximation, 
	there must exist a {\em target $c$-divider} $\tilde{\pi}$ which partitions $\tilde{V}$ into $\{\tilde{S}_{\tilde{\pi}}, \{\tilde{V}_{\text{left}}, \tilde{V}_{\text{right}}\}\}$ such that:
	\begin{itemize}
		\item each of $|\tilde{V}_{\text{left}}|$ and $|\tilde{V}_{\text{right}}|$ is at least $|\tilde{V}| (\frac{1}{4d + 4} - \frac{1}{r})$, and 	
		\item $|\tilde{S}_{\tilde{\pi}}| \leq 2c (1 + d)^{1/d}|\tilde{V}|(\frac{1}{|V|^{1/d}} + \frac{1}{r})$.
	\end{itemize}
	This is because at least $\pi^*$ is a target $c$-divider satisfying the above properties for $\tilde{V}$.

	Since $|\tilde{V}| \in O(s)$, it is stored in the local memory 
	of one machine.
 By examining all possible $c$-dividers in $\text{mbr}(V)$ for $\tilde{V}$, 
	at least one target $c$-divider $\pi$ 
	must be found ($\pi$ is not necessarily $\pi^*$), with which, by the property of $\frac{1}{r}$-approximation, we have the following properties of $\pi$ for $V$:
	\begin{itemize}
		\item $|V_{\text{left}}| \geq |V| (\frac{1}{4d + 4} - \frac{2}{r}) \geq \frac{|V|}{8(d+1)}$ and  $|V_{\text{right}}| \geq \frac{|V|}{8(d+1)}$,
		\item $|S_{\pi}| \leq 2c (1 + d)^{1/d}|{V}|(\frac{1}{|V|^{1/d}} + \frac{2}{r}) \leq 4c(1 + d)^{1/d}|V|s^{-1/d}$. 
	\end{itemize}
The lemma thus follows.
\end{proof}

Lemma~\ref{lmm:approx-bin-partition} suggests that, with only access to a $\frac{1}{r}$-approximation of $V$, we can compute a {\em good enough} $c$-divider $\pi$; it separates $V$ into 
two parts each with size $\geq \frac{|V|}{8(d + 1)}$; the size of the separator $S_{\pi}$ is still bounded by $O(\frac{c \cdot |V|}{s^{1/d}})$, the same as in Lemma~\ref{lmm:bin-partition}.

By a more careful analysis, interestingly, with a sample set $\tilde{V}$ of size $\Theta(s)$,
one can derive that $\tilde{V}$ indeed is sufficient to be, with high probability, a $\frac{1}{3r^{1 + l}}$-approximation, with 
some constant $l \in (0, \frac{1}{2})$.
The $\frac{1}{3r^{1 + l}}$-approximation property of $\tilde{V}$ is sufficient
to invoke Lemma~\ref{lmm:approx-bin-partition} for multiple times,
i.e., making multiple partitions, 
with $\tilde{V}$ in local memory.
We formalise this with the following crucial lemma.
\begin{lemmaapprep}\label{lmm:observation}
	Consider a $\frac{1}{3r^{1 + l}}$-approximation $\tilde{V}$ of the range space 
	$\Sigma = (V, \mathcal{R})$, where $|V| > r^l s$ and $l = \min\{\frac{1}{3}, \log_r \frac{|V|}{s}\}$ and $\mathcal{R}$ is the set of all possible axis-parallel rectangular ranges in $\text{mbr}(V)$. 
	Let $\B$ be an arbitrary range in $\mathcal{R}$.
	Define $V_1 = \Box \cap V$ and $\tilde{V}_1 = \Box\cap \tilde{V}$. 

	If $|\tilde{V}_1|\geq\frac{2 |\tilde{V}|}{r^l}$, 
	then we have: (i) $\tilde{V}_1$ is a $\frac{1}{r}$-approximation of $\Sigma'=(V_1,\mathcal{R}_1)$, where $\mathcal{R}_1=\{q \in \mathcal{R}\mid q\subseteq V_1\}$, and (ii) $|V_1| > s$.
Therefore, Lemma~\ref{lmm:approx-bin-partition} is applicable on $V_1$ and $\tilde{V}_1$. 
\end{lemmaapprep}
\begin{proof}
Consider an arbitrary rectangle $\Box'\subseteq \Box$. 
Denote $V'_1=\Box'\cap V=\Box'\cap V_1$ and $\tilde{V}'_1=\Box'\cap \tilde{V}=\Box'\cap \tilde{V}_1$. 
Since $\Box$ is an axis-parallel rectangular range and by the fact that $\tilde{V}$ is a $\frac{1}{3r^{1+l}}$-approximation of $\Sigma=(V,\mathcal{R})$, 
we have $\frac{|V_1|}{|V|}\geq\frac{|\tilde{V}_1|}{|\tilde{V}|}-\frac{1}{3r^{1+l}}$, and hence, 
\begin{equation}\label{eq:2}
	\frac{|V|}{|\tilde{V}|}\leq \frac{|V_1|}{|\tilde{V}_1|}+\frac{|V|}{|\tilde{V}_1|3r^{1+l}}\,. 
\end{equation}

On the other hand, $\Box'$ is also an axis-parallel rectangular range;
analogously, we have $\frac{|V'_1|}{|V|} \leq \frac{|\tilde{V}'_1|}{|\tilde{V}|}+\frac{1}{3r^{1+l}}$ and thus, 
\begin{equation}\label{eq:3}
	|V'_1|\leq \frac{|\tilde{V}'_1||V|}{|\tilde{V}|}+\frac{|V|}{3r^{1+l}}\leq |\tilde{V}'_1| \cdot (\frac{|V_1|}{|\tilde{V}_1|}+\frac{|V|}{|\tilde{V}_1|3r^{1+l}})+\frac{|V|}{3r^{1+l}} \,,
\end{equation}
where 
the second inequality is from Inequality~\eqref{eq:2}.

By the property of $\frac{1}{3r^{1 + l}}$-approximation and the fact that $|\tilde{V}_1|\geq\frac{2 |\tilde{V}|}{r^l}$, 
we have:
\begin{equation*}
\frac{|V_1|}{|V|} \geq \frac{|\tilde{V}_1|}{|\tilde{V}|} - \frac{1}{3r^{1 + l}} \geq (\frac{2}{r^l}-\frac{1}{3r^{1+l}}) >\frac{5}{3r^l}\,.
\end{equation*}

\noindent
And therefore, $|V_1| > |V| \cdot \frac{5}{3r^l} > s$; this proves Bullet (ii) in the lemma statement.

Finally, by the fact that $|\tilde{V}'_1| \leq |\tilde{V}_1|$ and substituting $|V_1| > |V| \cdot \frac{5}{3r^l}$ into Inequality~\eqref{eq:3}, we have $|V'_1|\leq \frac{|\tilde{V}'_1||V_1|}{|\tilde{V}_1|}+\frac{2|V|}{3r^{1+l}}\leq \frac{|\tilde{V}'_1||V_1|}{|\tilde{V}_1|}+\frac{2|V_1|}{5r} < \frac{|\tilde{V}'_1||V_1|}{|\tilde{V}_1|}+\frac{|V_1|}{r}$. 

Analogously, it can be shown that $|V'_1|\geq \frac{|\tilde{V}'_1||V_1|}{|\tilde{V}_1|}-\frac{|V_1|}{r}$.
Thus, $\tilde{V}_1$ is a $\frac{1}{r}$-approximation of $V_1$.
Lemma~\ref{lmm:observation} follows.
\end{proof}

\paragraph*{Our MPC Algorithm for Computing Pseudo $s$-Separator}

\subparagraph{One Super Round Implementation.}
Based on Lemma~\ref{lmm:observation}, given an implicit $(d,c)$-grid graph $G=(V, \mathcal{E})$, where $d\geq 3$, $|V| > r^l s$ for $l = \min\{\frac{1}{3}, \log_r \frac{|V|}{s}\}$, $1\leq c\leq s^{\frac{1}{d^3}}$ and $r = 2s^{1/d}$,
the implementation of one super round is shown in
Algorithm~\ref{algo:sep}.

\begin{algorithm}
	\caption{One Super Round Implementation for Pseudo $s$-Separator}\label{algo:oneround}
	\KwIn{an implicit $(d,c)$-grid graph $G=(V,\E)$ stored in a set of contiguous machines}
	\label{algo:sep}       
        Compute  a random sample set $\tilde{V}$ of $V$ in size $\Theta(s)$, send it to machine $M_0$\;

        \SetKwBlock{BeginBlock}{$\mathbf{M_0}$ processes locally:}{end}
\BeginBlock{
    Let $K \leftarrow \frac{2|\tilde{V}|}{r^l},\mathbb{V}\leftarrow\{\tilde{V}\}, \Pi\leftarrow \emptyset$\;
    \While{$\exists \tilde{V}_1\in \mathbb{V}$ s.t. $|\tilde{V}_1| \geq K$}{
    Apply Lemma~\ref{lmm:approx-bin-partition} on $\tilde{V}_1$;
    denote the resulting $c$-divider by $\tilde{\pi}$, which partitions $\tilde{V}_1$ into $\{\tilde{S}_{\tilde{\pi}}, \{\tilde{V}_{\text{left}}, \tilde{V}_{\text{right}}\}\}$ \;
    Let $\mathbb{V}\leftarrow\mathbb{V}\cup \{ \tilde{V}_{\text{left}}, \tilde{V}_{\text{right}}\}-\{\tilde{V}_1\}$, $\Pi\leftarrow\Pi \cup \{\tilde{\pi}\}$\;
}
}
        Broadcast $\Pi$ to all machines related to $V$\;
        \SetKwBlock{BeginxBlock}{Each machine processes locally:}{end}
        \BeginxBlock{
        Identify the separator vertices stored in their local memory with $\Pi$\;
        Group vertices according to the induced sub-graph, according to $\Pi$\;
        }
        After sorting, each sub-graph induced by $\Pi$ is stored in a set of contiguous machines\;
\end{algorithm}

\subparagraph{The Overall Pseudo $s$-Separator Algorithm.}
Given an implicit $(d,c)$-grid graph $G=(V, \mathcal{E})$, where $d\geq 3$, $|V| > r^l  s$ for $l = \min\{\frac{1}{3}, \log_r \frac{|V|}{s}\}$, $1\leq c\leq s^{\frac{1}{d^3}}$ and $r = 2s^{1/d}$,
\begin{itemize}
	\item perform a new super round \textit{simultaneously} on each of the induced sub-graphs $G'$ with $|V'| > r^{l'}  s$ for $ l' = \min\{\frac{1}{3}, \log_r \frac{|V'|}{s}\}$;
	\item terminate when
no more super rounds can be performed, in which case, 
all the induced sub-graphs have no more than $s$ vertices.
\end{itemize}

\begin{lemmaapprep}
	After a  super round, the input implicit $(d,c)$-grid graph $G = (V, \mathcal{E})$ is separated into $h$ 
	disconnected induced sub-graphs $G_i=(V_i, \mathcal{E})$ for $i \in [1, h]$
	such that $h \in O(r^l)$ and $|V_i| \in O(|V|/r^l)$ for all $i \in [1, h]$.
\end{lemmaapprep}
\begin{proof}
	According to the proof of Lemma~\ref{lmm:approx-bin-partition}, the sample set $\tilde{V}_i$ of each $G_i$ is 	
	of size at least $\frac{2|\tilde{V}|}{8(d+1)r^l}$ for all $i \in [1, h]$. 
	Therefore, there can be at most $O(r^l)$ applications of Lemma~\ref{lmm:approx-bin-partition} and hence, at most $h \in O(r^l)$ sub-graphs generated.
	By the fact that $\frac{|\tilde{V}_i|}{|\tilde{V}|} < \frac{2}{r^l}$ and the property of $\frac{1}{3r^{1 + l}}$-approximation $\tilde{V}$, we have: 
	$\frac{|V_i|}{|V|} < \frac{2}{r^l} + \frac{1}{3r^{1+l}} < \frac{3}{r^l}$, implying $|V_i| \in O(\frac{|V|}{r^l})$.
\end{proof}

\begin{lemmaapprep}\label{lmm:approx-analysis}
	When our Overall Pseudo $s$-Separator algorithm terminates, let $S$ be the set of all the vertices in $V$ falling 
	in target $c$-dividers in applications of Lemma~\ref{lmm:approx-bin-partition}.
	Set $S$ is a pseudo $s$-separator of the input implicit $(d,c)$-grid graph $G$.
\end{lemmaapprep}
\begin{proof}
The proof of Lemma~\ref{lmm:approx-analysis} is analogous to that of Lemma~\ref{lmm:s-sep-ext}, and thus omitted.
\end{proof}

\begin{lemmaapprep}\label{lmm:sep-round}
	The total number of super rounds of our overall pseudo $s$-separator algorithm is bounded by {$O(\frac{d}{\alpha})$}, and the total number of MPC rounds is bounded by {$O(1)$}.
\end{lemmaapprep}
\begin{proof}
	Since at the end of each super round, the vertex set size of each recursion instance is decreased by at least a factor of $\Omega(r^l)$. Hence, there can be at most {$O(\log_{r^l} \frac{|V|}{s/8(d+1)}) \subseteq  O(\frac{d}{\alpha})$ super rounds, because $l \geq \frac{1}{3}$, $r \geq 2s^{1/d}$ and $s = |V|^\alpha$.}
	Observe that each super round can be performed in {$O(\frac{1}{\alpha})$} MPC rounds.
Since both $d$ and $\alpha$ are constants, Lemma~\ref{lmm:sep-round} follows.
\end{proof}

\noindent
Putting the above lemmas together, we have:
\begin{theorem}\label{thm:sep-algo}
	Given an implicit $(d,c)$-grid graph $G=(V, \mathcal{E})$ with $d\geq 3$ and $1 \leq c \leq s^{\frac{1}{d^3}}$, 
	there exists a Las Vegas MPC algorithm with local memory {$\Theta(s) \subseteq \Theta(|V|^\alpha)$} per machine, for arbitrary constant $\alpha \in (0, 1)$, which computes a pseudo $s$-separator of $G$ in 
$O(1)$
rounds with high probability.
\end{theorem}

\subsection{Grid Graph Connectivity Algorithms}
\label{subsec:gridconnect}

Next, we introduce our MPC algorithms for computing CCs and MSFs for implicit $(d,c)$-grid graphs with $1 \leq c \leq s^{\frac{1}{d^3}}$.

\subparagraph{Bounding the Feasible Range for $\alpha$.}
Recall that Theorem~\ref{thm:sep-algo} shows that a pseudo $s$-separator $S$ of $G$ can be computed in $O(1)$ rounds with high probability.
This algorithm works for all constant  
$\alpha \in (0,1)$.
However, as we shall see shortly, both our CC and MSF algorithms require the pseudo $s$-separator, $S$, 
to fit in the local memory $O(s)$ in a single machine. 
As a result, 
they require $|S| \leq s$.
From the proofs of Lemma~\ref{lmm:approx-bin-partition} and Lemma~\ref{lmm:approx-analysis},
we know that $|S| \leq \frac{4c|V|}{s^{1/d}} \log s$.
Recalling that $s=|V|^{\alpha}$, our algorithm works when $\alpha\geq \frac{d + 1}{d+2}$ for $d\geq 3$. 
Combining the case of $d = 2$, a feasible range of $\alpha$ becomes $(\max\{ \frac{d+1}{d+2}, \frac{4}{5}\} , 1)$.

\begin{algorithm}[t]
	\caption{Connected Component Algorithm}\label{algo:cc}
	\KwIn{implicit $(d,c)$-grid graph $G=(V,\E)$ with $c\in[1,s^{\frac{1}{d^3}}]$}
        \KwOut{for each vertex $v\in V$, associate with the id of the connected component where it belongs}       
        calculate the pseudo $s$-separator with the algorithm in Section~\ref{subsec:pseudoseparator}\;
        \SetKwBlock{BeginBlock}{process each extended sub graph $G^+_i$ in their machine:}{end}

\BeginBlock{
    calculate the MSF of $G^+_i$, denote the edge set as $T_i$\;
    compress $T_i$ into an edge set $E'_i$, s.t. it induces a tree. All leaf nodes of the tree are in $S_i$, and the internal nodes of the tree are either from $S$ or have degree $> 2$\;
	send $E'_i$ to machine $M_0$\;
}
\SetKwBlock{BeginBlockb}{$\mathbf{M_0}$ run:}{end}
\BeginBlockb{
build a graph based on all of the received edges $\cup E'_i$\;
solve the CC problem on $E'_i$, get the CC id of vertices in $S$\;
send $S_i$ back to the corresponding machine along with their CC id's\;
}
\end{algorithm}

\subparagraph{Connectivity and Minimum Spanning Forest.} 
Consider a pseudo $s$-separator $S$ of $G$;
let $G_i = (V_i, \mathcal{E})$ be the induced sub-graphs of $G$ by $S$, where $i \in [1, h]$.
By Definition~\ref{def:s-sep}, we have $h = O(\frac{|V|}{s})$ and $|V_i| \in O(s)$.
For each $G_i$, 
denote the projection of $\text{mbr}(V_i)$ on dimension $j$ by $[x_i^{(j)}, y_i^{(j)}]$, for $i \in [1, h]$ and $j \in [1, d]$. 
Consider the hyper-box $H_i$ whose projection 
on each dimension $j$ is $[x_i^{(j)} - c, y_i^{(j)} + c]$ for $i \in [1, h]$ and $j \in [1, d]$;
we say that $H_i$ is the {\em extended region} of $G_i$. 
Moreover, we define the {\em extended graph} of $G_i$, denoted by $G_i^+ = (V_i^+, \mathcal{E})$, 
as the induced graph by the vertices in $V_i^+ = V_i \cup S_i$, where 
$S_i = S \cap H_i$ 
is the set of separator vertices falling in the extended region of $G_i$.
Figure~\ref{fig:grid-sep} shows an example of $G_i$ and $G_i^+$.
Clearly, by choosing a constant 
$\alpha \in (\max\{\frac{d+1}{d+2}, \frac{4}{5}\} , 1)$,
we know that $|S| \in O(s)$ and hence, $|V_i^+| \in O(s)$.
Therefore, $G_i^+$ fits in the local memory in one machine.
Moreover, without loss of generality, we assume that $G_i^+$ is stored completely in a machine $M_i$.

    \begin{algorithm}[t]
	\caption{MSF Algorithm}\label{algo:msf}
	\KwIn{implicit $(d,c)$-grid graph $G=(V,\E)$ with $c\in[1,s^{\frac{1}{d^3}}]$}
        \KwOut{calculate a MSF of $G$}
        calculate the pseudo $s$-separator with the algorithm in Section~\ref{subsec:pseudoseparator}\;
        \SetKwBlock{BeginBlock}{process each extended sub graph $G^+_i$ in their machine:}{end}

\BeginBlock{
    calculate the MSF of $G^+_i$, denote the edge set as $T_i$\;
    compress $T_i$ into an edge set $E'_i$ as described in Algorithm~\ref{algo:cc}\;
    //each edge $e$ in $E'_i$ corresponds to a path in $T_i$ and the heaviest edge weight in that path is called the {\em bottleneck edge of $e$} \;
    
    set the weight of edges in $E'_i$ the same as its bottleneck edge weight\;
	send $E'_i$ to machine $M_0$\;

}
\SetKwBlock{BeginBlockb}{$\mathbf{M_0}$ run:}{end}
\BeginBlockb{
build a graph based on all of the received edges $E'=\cup E'_i$\;
solve the MSF problem on $E'$, denoted as $T'$\;
assign to each edge in $E'$ a label to indicate whether it is in $T'$\;
send $E'_i$ back to the corresponding machine along with the label\;
}
\BeginBlock{
for each edge in $E'_i$ labeled as negative, i.e., not included in $T'$, (conceptually) remove its bottleneck edge from $G^+_i$ \;
//let $G_i^{+'}$ denote the extended sub graph after conceptually removing those edges\;
calculate the MSF of $G^{+'}_i$, denote the result as $\T_i$\;
}
//the union of all $\T_i$ is the MSF of $G$\;
\end{algorithm}

By incorporating the techniques of the existing External Memory CC and MSF algorithms for grid graphs~\cite{DBLP:journals/jgaa/GanT18}, 
where the detailed MPC implementations of them are respectively shown in 
Algorithm~\ref{algo:cc}  and Algorithm~\ref{algo:msf},
we have the following theorem:
\begin{theoremapprep}\label{thm:implicit-connectivity}
	Given an implicit (edge-weighted) $(d,c)$-grid graph $G=(V, \mathcal{E})$ with $1 \leq c \leq s^{\frac{1}{d^3}}$, 
there exists a Las Vegas MPC algorithm with local memory {$\Theta(s) \subseteq \Theta(|V|^\alpha)$} per machine, for arbitrary constant $\alpha \in (\max\{ \frac{d+1}{d+2}, \frac{4}{5}\} , 1)$, which computes all the connected components (resp., minimum spanning forest) of $G$ in 
$O(1)$ rounds with high probability.
\end{theoremapprep}
\begin{proof}
	The MPC algorithm is just an adaptation of the existing known CC and MSF algorithm for grid graphs~\cite{DBLP:journals/jgaa/GanT18}.
	Algorithm~\ref{algo:cc}  and Algorithm~\ref{algo:msf} give  the pseudo-code of our CC algorithm and the MSF algorithm, respectively.

	It can verified that Algorithm~\ref{algo:cc} and Algorithm~\ref{algo:msf} can be  performed in {$O(\frac{1}{\alpha})$} rounds.
\end{proof}

\section{$O(1)$-Round Approximate EMST}
\label{sec:emst}

\subparagraph*{An Overview of Our $\rho$-Approximate EMST Algorithm.}
The approach of our $\rho$-approximate EMST algorithm is 
to mimic the idea of the Kruskal's algorithm.
Specifically, our algorithm firstly 
uses those ``short'' edges to connect different connected components that are found so far, 
and then gradually considers using ``longer'' edges. 
Our algorithm
runs in super rounds, each of which takes $O(1)$ MPC rounds. 
In the $i$-th super round, 
our algorithm constructs a representative implicit $(d,c)$-grid graph $G_i = (V_i, \mathcal{E}_i)$ which only considers those edges with weight 
$\leq l_i = c^i(\frac{\rho}{\sqrt{d}})^{i-1}$, where $c = s^{\frac{1}{d^3}}$.
Next, it invokes our $O(1)$-round MSF algorithm on $G_i$,
and then starts a new super round repeating this process until a spanning tree $T$ is obtained.

\subparagraph{The Algorithm Steps.}
Consider a set $P$ of $n$ points in $d$-dimensional space $\mathbb{N}^d$ with coordinate value range $[0, \Delta]$.
Let $V_1 = P$ and set $c = s^{\frac{1}{d^3}}$.
Our MPC $\rho$-approximate EMST algorithm constructs a series of implicit $(d,c)$-grid graphs, 
denoted by $G_i = (V_i, \mathcal{E}_i)$, for $i \in \{1, 2, \ldots\}$.

\begin{itemize}
    \item \textbf{Augmented Node Information:} Each vertex $u$ in $V_i$ is associated with two 
 pieces of $O(1)$-size information: (i) a connected component id of $u$, denoted by $u.id$, and (ii) the original point $p\in P$ which $u$ currently represents, denoted by $u.pt$.
 \item \textbf{Edge Formation Rule:} for any two distinct nodes $u, v \in V_i$, in building, $\mathcal{E}_i$:
the edge $(u,v)$ exists {\em if and only if} the Euclidean distance $dist(u,v) \leq l_i$, where $l_i = c^i(\frac{\rho}{\sqrt{d}})^{i-1}$;
if it exists, its weight, 
$w(u,v)$, is
	$0$ if $u.id=v.id$, otherwise
	$\mathit{dist}( u,v)$.
\end{itemize}

Initially, $V_1 \leftarrow P$ and for each node $u \in V_1$, $u.id$ is the id of the point and $u.pt$ is the point itself.
Our $\rho$-approximate EMST algorithm runs in super rounds; in the $i$-th super round (for $i = 1, 2, \ldots$), it works as follows:
\begin{itemize}
	\item {\bf Solve Stage: MSF Computation on $G_i= (V_i, \mathcal{E}_i)$}:
		\begin{itemize}
			\item invoke our $O(1)$-round MPC algorithm (in Theorem~\ref{thm:implicit-connectivity}) to compute an MSF of $G_i$. In particular,
   let $u.cid$ be a temporary variable to denote the {\em current} connected component id of $u\in V_i$ in $G_i$. 
Observe the difference between $u.id$ and $u.cid$:  the former records the connected component id of $u$ in the {\em previous} super round, and is used to determine edge weights in $\mathcal{E}_i$ in the current round; 
			\item let $\mathcal{T}_i'$ be the set of edges in the MSF of $G_i$; remove edges with {\em zero} weight from $\mathcal{T}_i'$;
			\item initialize $\mathcal{T}_i \leftarrow \varnothing$;	for each edge $(x, y) \in \mathcal{T}_i'$, add an edge $(x.pt, y.pt)$ to $\mathcal{T}_i$; 
\item keep tract of the total edge number that are added to $T=\cup_i\T_i$; when, in total, $n-1$ edges are added, stop and return $T$ as an $\rho$-approximate EMST of $P$;
		\end{itemize}
	\item {\bf Sketch Stage: Computing  $V_{i+1}$ from $V_i$}:
		\begin{itemize}
			\item initialize $V_{i+1} \leftarrow \varnothing$; 
 impose a grid on $V_i$ with side length of $\frac{l_i\rho}{\sqrt{d}}$; 
   \item for each non-empty grid cell, $g$, add the most {\em bottom-left} corner point $t$ of $g$ to $V_{i+1}$; moreover, set $t.pt \leftarrow u.pt$ and $t.id \leftarrow u.cid$, where $u$ is  an arbitrary node in $V_i \cap g$;
		\end{itemize}
\end{itemize}

\subparagraph{MPC Implementation Details.} 
\begin{itemize}
    \item At the beginning of the Solve Stage, 
the vertex set $V_i$ is distributed across the machines. 
Thus, the MSF algorithm suggested in Theorem~\ref{thm:implicit-connectivity} is invoked on $G_i = (V_i, \mathcal{E}_i)$, 
and the resulted set of edges, $\mathcal{T}_i'$, is distributedly stored across the machines. 
By local computations, the edges in $\mathcal{T}_i'$ can be converted to edges in $\mathcal{T}_i$.
    \item To implement the Sketch Stage (i.e., computing $V_{i+1}$ from $V_i$), each machine first computes the grid cell coordinate for each vertex in $V_i$ locally. 
Our algorithm sorts all the vertices in $V_i$ by their grid cell coordinates.
After sorting, all the vertices falling in the same grid cell $g$ are either stored in just one machine or a contiguous sequence of machines.
Next, each machine generates the most bottom-left corner point $t$, according to our algorithm, for each grid cell $g$ in its local memory.
Duplicate points $t$ are then removed by Duplicate Removal, an atomic MPC operation whose implementation details can be found in Appendix~\ref{app:op}.
The vertex set $V_{i+1}$ for the next super round is thus constructed.   
\end{itemize}

\begin{figure}[t]
\label{fig:example}
  \begin{center}
  \includegraphics[height = 65mm]{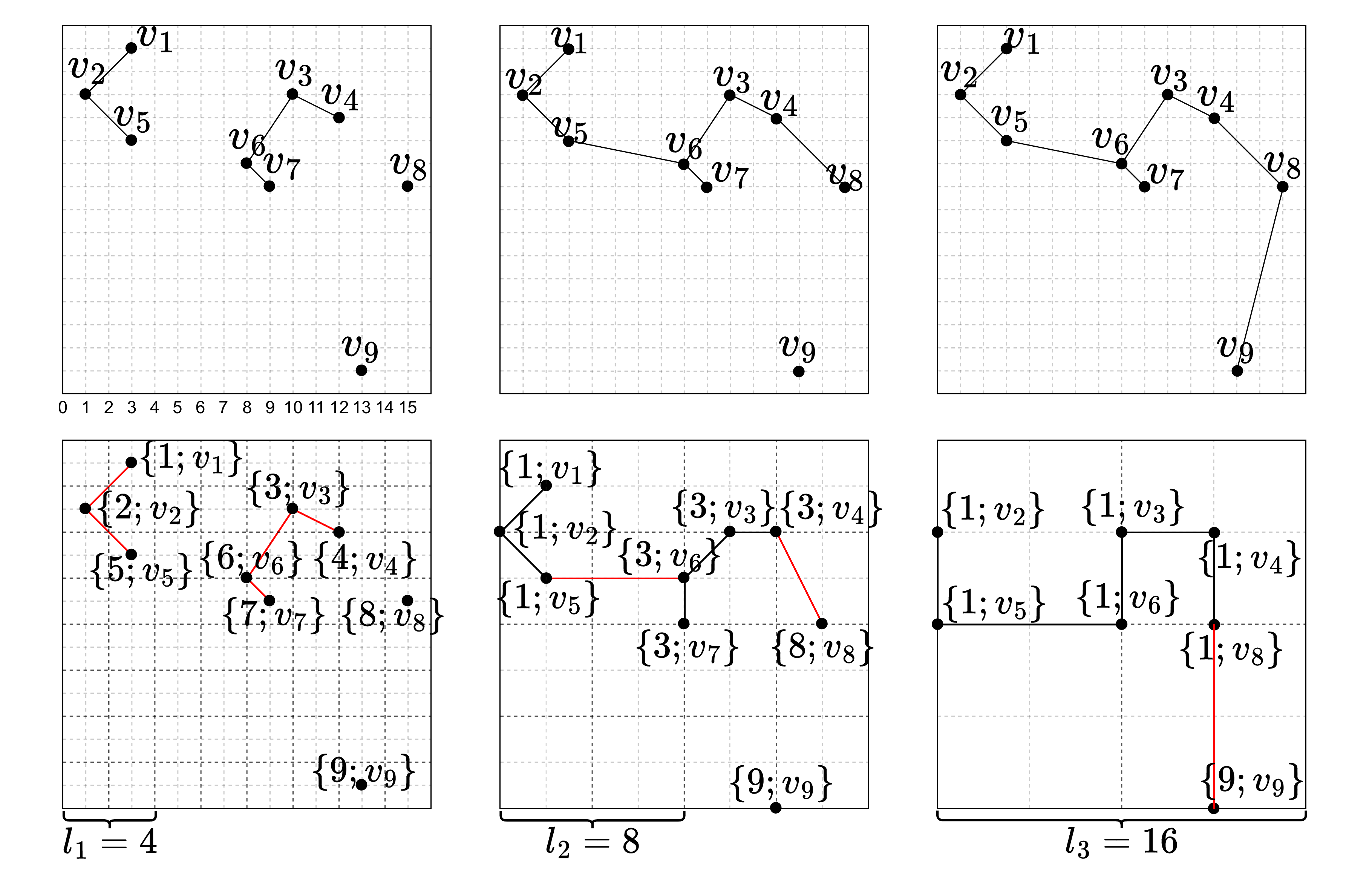}
  \end{center}

  \caption{ A running example of our approximate EMST algorithm
 }
  \label{fig:emst}

\end{figure}
\subparagraph{A Running Example.}
Figure~\ref{fig:emst} gives a running example of our approximate EMST algorithm. 
The input set $P=\{v_1,\ldots,v_9\}$ contains nine points in $2$-dimensional space with integer coordinate values in $[0,15]$. 
Our algorithm solves this problem in three super rounds. 
The top row of the figures shows the edges that are labeled to be in the result edge set $\cup_i \mathcal{T}_i$ at the end of the round $i$.
The bottom row shows the vertex set of $V_i$ and the MSF of $G_i$. 

In this example, $l_1=4$ and $\rho=\frac{\sqrt{2}}{2}$.
In the first round, our algorithm constructs $G_1=(V_1,\mathcal{E}_1)$. 
The vertex set $V_1$ with the augmented node information is shown in the first figure in at the bottom. 
The MSF of $G_1$, denoted as $\T'_1$, is computed and its edges are shown in the figure. 
After the computation of the MSF, 
each vertex $u\in V_1$ computes its component id $u.cid$.
All the zero-weight edges are deleted from $\T'_1$ and the remaining edges are highlighted in colour red. 
For each edge $(u_1,u_2)\in \T'_1$, its corresponding edge $(u_1.pt,u_2.pt)$ is marked as an output edge, i.e., the edge in the approximate EMST $T$ returned by our algorithm,  
as shown in the first top-row figure.
Next, the sketch stage starts; $G_2=(V_2,\mathcal{E}_2)$ is thus constructed as follows. 
A grid with side length of $\frac{l_1\rho}{\sqrt{d}} = 2$ is imposed on $V_1$, and then $V_2$ is generated, shown in the second bottom-row figure. 
At the end of this second round, two new edges are added to $T$, the final output edge set of the approximate EMST.
The same process is repeated with $l_2=8$ and $l_3=16$ in the third round, and  
one more edge is added to $T$. 
Since $T$ now contains in total eight edges, it is returned as an approximate EMST on the input nine points.

\subparagraph{Remark on $G_i = (V_i, \mathcal{E}_i)$.} 
From the construction of $G_i$ in our algorithm, 
at first glance, $G_i$ is an implicit $(d, \frac{c \cdot l_i \rho}{\sqrt{d}})$-grid graph which looks ineligible for our MSF algorithm proposed in Theorem~\ref{thm:implicit-connectivity}.
However, observe that since the coordinate values of the nodes in $V_{i}$ are all multiples of  
$\frac{l_i\rho}{\sqrt{d}}$, therefore, by a simple scaling, $G_i$ is indeed an implicit $(d, c)$-grid graph.

\begin{lemmaapprep}
Our $\rho$-approximate EMST algorithm takes $O(1)$ MPC rounds with high probability.
\end{lemmaapprep}
\begin{proof}
There are at most $O(\frac{\log \Delta}{\log (s^{1/d^3}\rho/\sqrt{d})})$ super rounds.
In each super round, the invocation of our MSF algorithm for a $(d,c)$-grid graph takes $O(1)$ MPC rounds with high probability (by Corollary~\ref{cor:grid-msf}), and the construction of the next $G_i(V_i, \mathcal{E}_i)$ can be achieved by sorting and scanning the points, which is $O(1)$-MPC-round doable.
Therefore, our $\rho$-approximate EMST algorithm takes $O\left(\frac{d\log \Delta}{\alpha^2\log \left(s^{1/d^3}\rho/\sqrt{d}\right) }\right)\subseteq O\left (\frac{d^4}{\alpha^3}\right ) \subseteq O(1)$ rounds.
\end{proof}

Let $T$ be the union of $\mathcal{T}_i$ for all $i = 1, 2, \ldots$. 
We prove that $T$ is an $(2\rho)$-approximate EMST of $P$, achieving both overall and edge-wise approximation.
By decreasing the constant approximation parameter $\rho$ by half, an $\rho$-approximate EMST can be obtained.

\begin{toappendix}
For the ease of discussion, we first introduce some frequently used symbols and notation.
We use $i$ to denote the super-round number. 
Given a point $v\in P$, 
$g^{(i)}(v)$ denotes 
the grid cell that contains $v$ with side length of $\frac{l_i\rho}{\sqrt{d}}$. 
Observe that, according to the way that $V_{(i+1)}$ is constructed from $V_{i}$ in our algorithm, 
we have $g^{(i)}(v) \subseteq g^{(i+1)}(v)$. 
Moreover, we denote the most bottom-left corner point of $g^{(i)}(v)$ by $v^{(i+1)}$, which indeed is a point in $V_{i+1}$. 
Denote the union of all the tree edges found up to the end of the $i$-th super round by $T_i=\cup_{k\in[1,i]}\T_k$.
\end{toappendix}

\begin{lemmaapprep}\label{lmm:spanning-tree}
$T$ is an Euclidean spanning tree of $P$.
\end{lemmaapprep}
\begin{proof}
We  prove the following claim with mathematical induction:
\begin{claim}\label{claim:no-cycle}
	At the end of super round $i$, 
	\begin{itemize}
		\item (i) any two points $u,v\in P$ in the same grid cell with side length of $l_{i}\rho/\sqrt{d}$ are connected in $T_i$, and 
		\item (ii) there is no cycle in $T_i$.  \lipicsEnd
	\end{itemize}
\end{claim}

When $i=1$, the claim holds trivially. Suppose for all $i \leq k$, the claim holds. 
Next, we consider the case of $i = k + 1$.
\begin{itemize}
	\item  Let $u,v\in P$ denote two points in the same grid cell in the $(k+1)$-th super round, i.e., $g^{(k+1)}(u)=g^{(k+1)}(v)$.
		Recall that $u^{(k+1)}$ and $v^{(k+1)}$ are the points representing $u$ and $v$ in $V_{i+1}$, respectively, and they are essentially the most left-bottom points of the cells $g^{(k)}(u)$ and $g^{(k)}(v)$, respectively.

If $u^{(k+1)} = v^{(k+1)}$, then $u$ and $v$ are already in the same grid cell in the $k$-th round, i.e., $g^{(k)}(u) = g^{(k)}(v)$.
By the induction assumption, $u$ and $v$ are already connected in $T_k$, and hence in $T_{k+1}$ as $T_k \subseteq T_{k+1}$.
   
Otherwise, since $u$ and $v$ 
are in the same grid cell imposing on $G_{k+1}$, 
equivalently, $u^{(k+1)}$ and $v^{(k+1)}$ are in the same grid cell.
Thus, 
the distance between $u^{(k+1)}$ and $v^{(k+1)}$ 
is at most $\rho l_{k+1}$, and hence, by the edge formation rule $\mathcal{E}_{k+1}$, 
the edge between $u^{(k+1)}$ and $v^{(k+1)}$ is considered when computing the MSF on $G_{k+1}$.
Therefore, $u^{(k+1)}.pt$ and $v^{(k+1)}.pt$ must be connected in $T_{k+1}$.
On the other hand, by the induction assumption, 
$u$ and $u^{(k+1)}.pt$ are connected in $T_k$ because they are in the same grid cell in the $k$-th super round, and for the same reason,
$v$ and $v^{(k+1)}.pt$ are connected in $T_k$.
Therefore, $u$ and $v$ are connected in $T_{k+1}$. 
This completes the proof of Bullet (i) of Claim~\ref{claim:no-cycle}.

\item Next, we prove $T_{k+1}$ has no cycle by contradiction.
Suppose there is a cycle in $T_{k+1}$; since, by the induction assumption, we know that $T_k$ has no cycle,
thus, there must be a set $E_{\text{bad}}$ of edges in $\mathcal{T}_{k+1}$ adding which to $T_k$ introduces a cycle.
There are two cases of this cycle:
\begin{itemize}
	\item Case 1: this cycle is formed just within one connected component $C$ (which is a tree) of $T_k$.
		It implies that there is an edge $(u,v)$ in $\mathcal{T}_{k+1}$ where $u$ and $v$ are in the same connected component $C$ in $T_k$.
		In other words, $u^{(k+1)}$ and $v^{(k+1)}$ must have the same connected component id before the computation of MST on $G_{k+1}$ and hence, the edge between them must have zero weight. 
		Therefore, $(u^{(k+1)}, v^{(k+1)})$ must not be in $\mathcal{T}'_{k+1}$ (which has no zero-weight edges)  and thus, $(u,v)$ must not be in $\mathcal{T}_{k+1}$. 
		Contradiction.

	\item Case 2: this cycle is formed across a sequence of connected components $C_1, C_2, \ldots, C_q$ in $T_k$, such that $C_1$ and $C_q$ are connected and any two consecutive components $C_j$ and $C_{j+1}$ are connected.	
		And all the edges connecting them must come from $\mathcal{T}_{k+1}$ because $T_k$ has no cycle.
		Consider the corresponding edges in $\mathcal{T}'_{k+1}$  of these edges in $\mathcal{T}_{k+1}$.
		Denote the set of these edges in order as $\{(u_j^{(k+1)}, v_j^{(k+1)})\}$, where $u_j^{(k+1)}, v_j^{(k+1)} \in V_{k+1}$ for all $j\in [1, q]$.
		Thus, before the computation of MSF on $G_{k+1}$, 
		we must have that 
the connected component id associated with the ``right endpoint'' of an edge must be the same as the connected component id associated to the ``left endpoint'' of the next edge in the sequence, i.e., 
		$v_j^{(k+1)}.id = u_{j+1}^{(k+1)}.id$. 
		Therefore, according to the edge formation rule $\mathcal{E}_{k+1}$, 
		all these endpoints must be connected with zero-weight edges in the MSF of $G_{k+1}$.
		In other words, there will be a cycle in the MSF of $G_{k+1}$ which is a contradiction.
\end{itemize}
Therefore, in any case, $E_{\text{bad}}$ does not exist and therefore, $T_{k+1}$ contains no cycle.
This completes the proof of Bullet (ii) of Claim~\ref{claim:no-cycle}.
\end{itemize}

Thus, Claim~\ref{claim:no-cycle} holds for all super rounds.
When the last super round ends, 
all the points are connected in $T$ and there is no cycle. 
So $T$ is an Euclidean spanning tree of $P$.
This completes the proof of Lemma~\ref{lmm:spanning-tree}.
\end{proof}

\begin{lemmaapprep}[Edge-Wise Approximation Guarantee]
\label{lem:edgewiseproof}
	Consider an exact EMST $T^*$ of $P$; for any edge $(u, v) \in T^*$, there must exist a path $\mathcal{P}(u,v)$ from $u$ to $v$ in $T$ such that every edge $(x,y)$ on this path has weight $w(x,y) \leq (1 + 2\rho) \cdot w(u,v)$.
\end{lemmaapprep}

\begin{proof}
	For any edge $(u,v) \in T^*\setminus T$, that is, $(u,v)$ is in $T^*$ but not in $T$, 
	consider the path from $u$ to $v$ in $T$, denoted by $\P(u,v)$. 
	Furthermore, let $i$ be super round number, where 
	$u$ and $v$ are connected for the {\em first time} in $T_i$; recall that $T_i$ is the sub-tree of $T$ computed up to the end of the $i$-th super round.
	Observe that this implies $i>1$, because otherwise $T$ must also contain $(u,v)$ according to computation in the first super round of our algorithm. 
	Thus, we have: (i) $u^{(k)}.pt=u,v^{(k)}.pt=v$ holds for all $k\leq i$, (ii) 
	$u,v$ are not connected in $T_{i-1}$. 
	In other words, in the $(i-1)$-st super round, 
	$u^{(i-1)}$ and $v^{(i-1)}$ are not connected in $G_{i-1}$.
	It further implies that $dist(u^{(i-1)},v^{(i-1)})>l_{i-1} \Rightarrow dist(u,v)>l_{i-1}-\rho l_{i-2}$, which gives a lower bound on $dist(u,v)$.

	In the following, we categorize the edges in $\P(u,v)$ by the super round number $1 < k \leq i$ when they are added in $T_i$.

\begin{itemize}
    \item \textbf{Category 1: $k = i$.} For any edge $(x,y)$ in $\T_i\cap \P(u,v)$, 
$dist(u^{(i)},v^{(i)}) \geq dist(x^{(i)},y^{(i)})$ holds, 
because otherwise, $(u,v) \in T_i$ contradicting the fact that $(u,v) \in T^*\setminus T$. 
By the triangle inequality, we have $dist(x,y)\leq dist(u,v) + \rho l_{i-1}$ 
Combining this inequality with the lower bound on $dist(u,v)$, 
we have $\frac{dist(x,y)}{dist(u,v)}\leq 1+2\rho$. 
\item \textbf{Category 2: $k < i$.} For any edge $(x,y)$ in $\T_k\cap \P(u,v)$, where $k\leq i-1$, we have $dist(x^{(k)},y^{(k)})\leq l_k\Rightarrow dist(x,y)\leq l_{i-1}+\rho l_{i-2}$. Likewise, $\frac{dist(x,y)}{dist(u,v)}\leq 1+2\rho$ holds.
\end{itemize}

Therefore, for every edge $(x,y)$ on $\P(u,v)$, $dist(x,y) \leq (1 + 2\rho) \cdot dist(u,v)$ holds, and the lemma follows.
\end{proof}

\begin{lemmaapprep}[Overall Approximation Guarantee]\label{lmm:emst-overall}
	Consider an exact EMST $T^*$ of $P$; the total edge weight of $T$ is at most $(1 + 2\rho)$ times the total edge weight of $T^*$. 
\end{lemmaapprep}
\begin{proof}
We sort the edges in $T^{*}$ in ascending order w.r.t. their weights, $e^{*}_1, e^{*}_2,\ldots , e^{*}_{n-1}$. 
For each $e^{*}_i=(u_i,v_i)$, we denote $\P_i$ denote the path in $T$ from $u_i$ to $v_i$.
Let $E_i=\cup_{k\leq i}\P_k$ for $i\in [1,n-1]$ denote the union of the edges that formed by those paths. 
So clearly $|E_i|\geq i$ and for any edge $e\in E_i$, we have $w(e)\leq (1+2\rho)\cdot w(e^*_i)$ according to Lemma~\ref{lem:edgewiseproof}.
We can show the following claim: 
\begin{claim}
For any $1\leq i\leq n-1$, we can always find a subset $E'_i\subseteq E_i$ that $|E'_i|=i$ and $\sum_{e\in E'_i} w(e)\leq (1+2\rho)\sum_{k\leq i} w(e^*_k)$. \lipicsEnd
\end{claim}

This can be proved with mathematical induction. For $k=1$ case, it holds trivially by putting an arbitrary edge in $E_1$ into $E'_1$. 
For $i\geq 2$ cases, suppose the claim holds true for $i=k$, which means we can find $E'_k$ confirms the claim. 
We know that $|E_{k+1}|\geq k+1$ and $|E'_k|=k$.
Let $e_{k+1}$ be an arbitrary edge in $E_{k+1}-E'_k$.
According to Lemma~\ref{lem:edgewiseproof},
$w(e_{k+1})\leq (1+2\rho)\cdot w(e^*_{k+1})$ holds.
We define $E'_{k+1}=E'_k\cup \{e_{k+1}\}$, which proves the claim.

When $i=n-1$, $E_{n-1}=T$ and the claims indicates that $T$ achieves the overall approximation guarantee.

This completes the proof of Lemma~\ref{lmm:emst-overall}.

\end{proof}

\vspace{-1mm}
\section{$O(1)$-Round Approximate DBSCAN}
\label{sec:dbscan}

\subsection{DBSCAN and $\rho$-Approximate DBSCAN}

\subparagraph{Core and Non-Core.}
Let $P$ be a set of $n$ points in $d$-dimensional Euclidean space $\mathbb{R}^d$, where the dimensionality $d \geq 2$ is a constant.
For any two points $u, v \in P$, the Euclidean distance between $u$ and $v$ is denoted by $\dist(u,v)$. 
There are two parameters: (i) a radius $\eps \geq 0$ and (ii) core point threshold $\minPts \geq 1$ which is a constant integer.
For a point $v \in P$, define $B(v, \eps)$ as the ball center at $v$ with radius $\eps$. 
If $B(v,\eps)$ covers at least $\minPts$ points in $P$, then $v$ is a {\em core} point. 
Otherwise, $v$ is a {\em non-core} point. 

\subparagraph{Primitive Clusters and Clusters.}
Consider an implicit {\em core graph} $G = (P_{\text{core}}, \mathcal{E})$, 
where $P_{\text{core}}$ is the set of all the core points in $P$, 
and there exists an edge between two core points $u$ and $v$ if $\dist(u,v) \leq \eps$.
Each connected component $C$ of $G$ is defined as a {\em primitive cluster}.   
For a primitive cluster $C$, let $C'$ be the set of all the non-core points $p \in P$ such that $\dist(p, C) = \min_{u \in C} \dist(p, u) \leq \eps$. The union of $C \cup C'$ is defined as {\em DBSCAN cluster}.

\subparagraph{The DBSCAN Problem.} Given a set of points $P$, a radius $\eps \geq 0$ and a core threshold $\minPts \geq 1$, 
the goal of the DBSCAN problem is to compute all the DBSCAN clusters. 

\subparagraph{$\rho$-Approximate DBSCAN.}
Given an extra approximation parameter, $\rho \geq 0$, 
the only difference is in the definition of 
the primitive clusters.
Specifically, consider a conceptual {\em $\rho$-approximate core graph} $G_\rho = (P_{\text{core}}, \mathcal{E}_\rho)$,
where 
the edge formation rule $\mathcal{E}_\rho$ works as follows. 
For any two core points $u, v \in P_{\text{core}}$:
if $\dist(u,v) \leq \eps$, $\mathcal{E}_\rho(u,v)$ must return that $(u,v)$ is in $G_\rho$;
if $\dist(u,v) > (1 + \rho) \eps$, $\mathcal{E}_\rho(u,v)$ must return that $(u,v)$ does not exist in $G_\rho$;
otherwise, what $\mathcal{E}_\rho(u,v)$ returns does not matter.

Due to the does-not-matter case, 
$G_\rho$ is not unique.
Consider a graph $G_\rho$; each connected component $C$ of $G_\rho$ is defined as a primitive cluster with respect to $G_\rho$.
And an $\rho$-approximate DBSCAN cluster is obtained by including non-core points to a primitive cluster $C$ in the same way as in the exact version.
The goal of the $\rho$-approximate DBSCAN problem is to compute the $\rho$-approximate DBSCAN clusters with respect to an arbitrary $G_\rho$.

\subsection{Our MPC Algorithm}
\label{sec:ourMPCalgos}
Our approximate DBSCAN algorithm consists of three main steps: 
(i) identify core points; (ii) compute primitive clusters from a $G_\rho$; and (iii) assign  non-core points to primitive clusters. As shown in the existing $\rho$-approximate DBSCAN algorithm in~\cite{gan2017hardness},
the third step, assigning non-core points, can be achieved by an analogous algorithm of core point identification.
Next, we give an MPC algorithm for the first two steps.

\subsubsection{Identify Core Points}
In this section, 
we show how to identify all the core points for a DBSCAN problem instance in $O(\frac{1}{\alpha})$ MPC rounds.
Our goal is to 
label each point $p\in P$, as to whether or not it is a {\em core point}.
In the following, 
we assume that the constant parameter $\alpha$ in the MPC model satisfies~$\alpha \in (\frac{1}{2},1)$, and hence, $m\leq s$, that is, the number of machines is no more than the local memory size $s$.

Consider a grid imposed in the data space $\mathbb{R}^d$ with grid-cell length 
$l=\varepsilon/\sqrt{d}$.
Each grid cell,~$g$, is represented by the coordinates $(x_1, x_2, \ldots, x_d)$ of its {\em left-most corner}, namely, each grid cell $g$ is essentially a hyper cube $[x_1l,(x_1+1)l)\times\cdots\times[x_dl,(x_d+1)l)$. 
Each point $p\in P$ is contained in one and exactly one grid cell, denoted by $g(p)$. 
The size of a grid cell $g$, denoted by $|g|$, is defined as the number of points in $P$ contained in $g$.
If $|g| >0$, then $g$ is a {\em non-empty} grid cell.
If a non-empty grid cell $g$ contains at least $\minPts$ points of $P$, 
then $g$ is a {\em dense grid cell}.
Otherwise, $g$ is a {\em sparse grid cell}.
Given a non-empty grid cell $g$, the neighbour grid cell set of $g$ is defined as $\Nei_g=\{g' \neq g\mid \exists p \in g, q\in g', s.t.\ \dist(p,q)\leq\varepsilon \}$ (note here that $p$ and $q$ are not necessarily points in $P$). 
A grid cell, $g'$, in $\Nei_g$ is called a neighbour grid cell of $g$. 
It can be verified that $|\Nei_g| < (2\sqrt{d}+3)^d \in O(1)$~\cite{gan2015dbscan,DBLP:journals/jgaa/GanT18}.

\smallskip
\noindent
Our core point identification algorithm contains the following three steps:

\subparagraph{Put Points into Grid Cells.} 
Given a point~$p$, the coordinates of the grid cell~$g(p)$ which contains $p$ can be calculated with the coordinates of~$p$. 
Define $P_g= P \cap g$ as the set of all the points in $P$ falling in $g$.
By sorting all the points in $P$ by the coordinates of the grid cells they belong to in lexicographical order,
which can be achieved in $O(\log_s n)$ MPC rounds~\cite{DBLP:conf/isaac/GoodrichSZ11}, 
all the points in $P$ can be arranged in the machines with the following properties:
\begin{itemize}
    \item Each machine stores points from consecutive 
    (in the sorted order) non-empty grid cells.
    Those non-empty grid cells are called the \textit{support 
    grid cells} of the machine. 
	\item The points in a same grid cell $g$ are stored in either only one machine or a contiguous sequence of machines. 
Thus, the ID's of the machines that store~$P_g$ --  the storage locations of the grid cell~$g$ -- 
constitute an interval and 
can be represented by the left and right endpoints (i.e., the smallest and largest machine ID's) of this interval in~$O(1)$ words.
	\item The number of the support grid cells of each machine is bounded by $O(s)$. And the total non-empty grid cell count is bounded by~$O(n)$.
\end{itemize}

\subparagraph{Calculate Non-empty Neighbour Grid Cell Storage Locations.} 
In this step, we aim to, for each non-empty grid cell, calculate the locations of all its {\em non-empty} neighbour grid cells. 
Since the size of $\Nei_g$ is bounded by $O((2\sqrt{d}+3)^d) \in O(1)$, 
it takes $O(1)$ space to store the location intervals of its all (possibly empty) neighbour grid cells. 
The pseudo code of the following process are shown in Algorithm~\ref{algo:neiloc}.

At high level, our algorithm takes four {\em super rounds} to calculate non-empty neighbour grid cell storage locations, as follows:
\begin{itemize}
\item Super Round 1: For each support grid cell $g_j$ in a machine $M_i$, for each of its possible neighbours $g'$, 
our algorithm builds a  4-tuple $(g_j.\coord,M_i, g'.\coord,g'.\loc)$, 
where $g.\coord$ denotes the coordinates of grid cell $g$, $M_i$ denotes the machine ID and $g'.\loc$ term is a placeholder with value NULL at the current stage and will be filled with the storage location of $g'$ if $g'$ is non-empty.  
Next, these 4-tuples are \underline{sorted} by the third term,~$g'.\coord$, across all machines.
Therefore, at the end of this super round, each machine stores a subset of these $4$-tuples in its memory.
\item Super Round 2: Each machine $M$ \underline{broadcasts} the minimal and maximal $g'.\coord$ value of the $4$-tuples it stores, denoted by $[M.\min(g'.\coord), M.\max(g'.\coord)]$.
As a result, at the end of this super round, each machine, $M_i$, has the information of the $g'.\coord$ value range of each other.
\item Super Round 3:  Each machine $M_i$ \underline{sends} a $2$-tuple 
$(g_j.\coord,g_j.\loc)$, for each of its support grid cells $g_j$, to those machines $M$ having $g_j.\coord \in [M.\min(g'.\coord),M.\max(g'.\coord)]$. 
Specifically, the second term of the $2$-tuple, $g_j.\loc$, is the storage location of $g_j$, which was \emph{computed} in the previous ``put points into grid cells'' phase.
\item Super Round 4:
After receiving all $g_j.\loc \in [M.\min(g'.\coord), M.\max(g'.\coord)]$ in Super Round 3, 
the previous placeholder, the forth term $g'.\loc$, in each $4$-tuple stored in machine $M$ now can be filled properly in $M$'s local memory.
Each of the resultant $4$-tuples $(g_j.\coord, M_i, g'.\coord, g'.\loc)$ are then \underline{sent} back to machine $M_i$ accordingly.
\end{itemize}

\begin{algorithm}[t]
	\caption{Location of Non-empty Neighbor Grid Cell}\label{algo:neiloc}
\KwIn{A set of $m$ machines $(M_1,\ldots,M_m)$, each machine~$M_i$ contains points from $k^{(i)}$ grid cells $(g_1^{(i)},\cdots,g_{k^{(i)}}^{(i)})$, the storage location of the grid cell $g.\loc$ }
\KwOut{Given a non-empty grid cell~$g$, calculate the storage location of every $g'\in \Nei_g$}
\For{machine $M_i$, $i\in[1,m]$ simultaneously}{
\For{each support grid cell~$g_j$, $j\in [1,k^{(i)}]$}{
    build $4$-tuples $(g_j.\coord,M_i, g'.\coord,g'.\loc)$ for each $g'\in \Nei_{gj}$\;
    //the fourth term is a placeholder with value NULL now
}
}
invoke a sorting algorithm for all of the $4$-tuples based on the $g'.\coord$\;
\For{machine $M_i$, $i\in[1,m]$ simultaneously}{
    broadcast the minimal and maximal $g'.\coord$ \;
    //each machine also receives that information during the communication phase\;
    \For{each grid cell $g_j$, $j\in [1,k^{(i)}]$}{
    send the $2$-tuple $(g_j.\coord,g_j.\loc)$ to the machines $M$ s.t. $g_j.\coord \in [M.\min(g'.\coord),M.\max(g'.\coord)]$\;
}
}
//each machine also receives the 2-tuple during the communication phase\;
merge the 2-tuple with 4-tuple to fill in the placeholder\;
\For{each 4-tuple with $g'.\loc!=\text{NULL}$}{
    send the 4-tuple back to $M_i$\;
}
//each machine receive the backward 4-tuple during the communication phase\;
\end{algorithm}

\begin{lemmaapprep}
\label{lem:sendsizebound}
The above algorithm, i.e., Algorithm~\ref{algo:neiloc}, takes  
$O(1)$ rounds. In each  round, the amount of data sent and received by each machine is bounded by $O(s)$.
\end{lemmaapprep}
\begin{proof}
Clearly, there are four super rounds in Agorithm~\ref{algo:neiloc}.
Furthermore, it can be verified that each super round can be implemented with $O(1)$ atomic MPC operations, e.g., sorting and broadcasting. As shown in Appendix~\ref{app:op}, each of these MPC operations can be performed in $O(1)$ rounds. Thus, the total round number is bounded by $O(1)$.

Next, we analyse the amount of data communicated in each round one by one.

Since each machine only contains at most~$O(s)$ support grid cells, 
and 
each support grid cell has $O(1)$ neighbour grid cells, 
the number of $4$-tuples that a machine generates is bounded by~$O(s)$.

Broadcasting the minimal and maximal coordinates of~$g'$ in a machine takes~$O(m)$ outgoing words and~$O(m)$ incoming words. By the fact that $\alpha\geq \frac{1}{2}$, we have $O(m) \subseteq O(s)$.

For the step of sending the $2$-tuples, we first analyse the incoming message size for each machine. 
For machine~$M$, denoted by $\mathcal{G}_M$ the non-empty grid cells with coordinate falling in $M$'s $g'.\coord$ range. Each grid cell in $\mathcal{G}_M$ implies that $M$ will receive a $O(1)$-word message from the machine that stores $g$, i.e., $g.\loc$. 
Therefore, the total number of $O(1)$-word messages received by~$M$ is bounded by 

		    \begin{equation*}
	    \sum_{g\in \mathcal{G}_M}|g.\loc|\leq |\mathcal{G}_M|+m-1 \in O(s)\,. 
    \end{equation*}

As for the total number of the outgoing messages, 
if machine~$M$ contains~$k$ grid cells, for each grid cell~$g_i$, for $i\in[1,k]$, it sends a 2-tuple to a set $\mathcal{M}_i$ of machines.
Since these 2-tuples are sorted by the~$g'.\coord$, 
each~$\mathcal{M}_i$ contains a contiguous set of machines and the total number of outgoing $O(1)$-word messages is bounded by 
		    \begin{equation*}
	    \sum_{i\in[1,k]}|\mathcal{M}_i|\leq k-1+m \in O(s)\,. 
    \end{equation*}

The final step is to send the 4-tuples to where it was generated, so the total number of the outgoing and incoming $O(1)$-word messages are bounded by~$O(s)$.
\end{proof}

\subparagraph{Output the Point Labels.} 
Algorithm~\ref{algo:coredete} gives the pseudo code of the detailed implementation. 
A running example of Algorithm~\ref{algo:coredete} is shown in Figure~\ref{fig:dbcore}.
This algorithm labels each point in $P$ as whether or not it is a core point. 
If a grid cell $g$ is dense, i.e., $|g| \geq \minPts$, 
all of the points in $P_g$ are labelled as core points. 
For each point $p$ in each sparse grid cell $g(p)$ (i.e., with $|g(p)| < \minPts$) in each machine $M$, 
we can count the number of the points of $P$ in the hyper ball centred at $p$ with radius $\varepsilon$, denoted by $B(p, \varepsilon)$. 
Observe that a point $q$ falls in $B(p,\varepsilon)$ only if either $p$ and $q$ are in the same grid cell or $g(q)$ is a neighbour grid cell of $g(p)$. 
Our algorithm, thus, sends $p$ to all the storage locations of (i.e., the machines storing) the neighbour grid cells of $g(p)$.
In the local memory of each of those machines, 
the distances between $p$ and all those points in the neighbour grid cells can be computed.
And then, the numbers of points falling in $B(p, \varepsilon)$ in those machines are sent back to machine $M$ (where $p$ is stored).
Therefore, the points $p$ in sparse cells can be labelled accordingly.

\begin{algorithm}[t]
\caption{Output the Point Label}\label{algo:coredete}
\KwIn{A set of $m$ machines $(M_1,\cdots,M_m)$, machine $M_i$ contains points from $k^{(i)}$ grid cells $(g_1^{(i)},\cdots,g_{k^{(i)}}^{(i)})$, the storage location of the grid cell $g.\loc$}
\KwOut{For each point $p$, assign to it a label as core or non-core}

\For{machine $M_i$, $i\in[1,m]$ simultaneously}{
\For{each grid cell $g_j$, $j\in [1,k^{(i)}]$}{
    \If{$g_j$ is a dense grid cell}{
        $p.\type\gets \text{core}$\;
    }\Else{
        \For {each point $p$ in $g_j$}{
        \For{each non-empty neighbor $g'$ of $g_j$}{
            send $p.\coord$ to all machines in $g'.\loc$\;
            }
        }
    }
}
}
//each machine receives a set of points $P_{in}$\;
\For{machine $M_i$, $i\in[1,m]$ simultaneously}{
\For{each point $p$ in $P_{in}$}{
    count the point number that falls in $B(p,\varepsilon)$\;
    send that number back to where point $p$ is from\;
}
//each machine receives the number\;
\For{each point $p$ in sparse grid cell}{
    count the total point number that falls in $B(p,\varepsilon)$\;
    identify the type of $p$
}
}
\end{algorithm}

\begin{figure}[t]
  \begin{center}
  \includegraphics[width=1.0\textwidth]{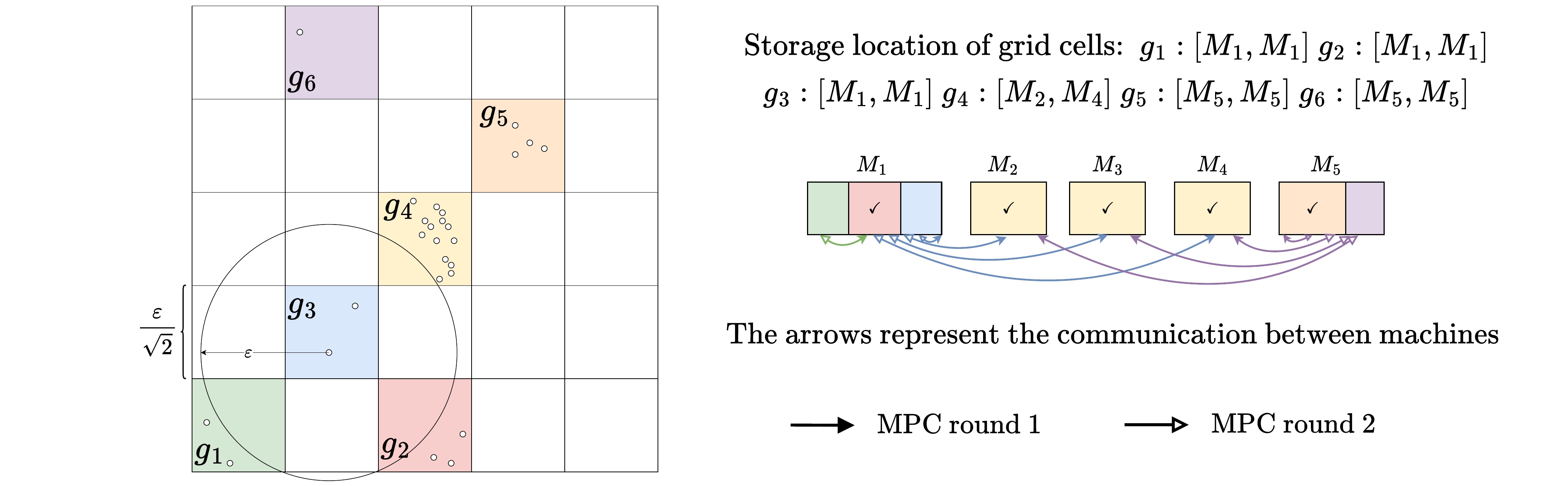}
  \end{center}

  \caption{ A running example of Algorithm~\ref{algo:coredete} with $\minPts=3$. The input points are in six non-empty grid cells stored across five machines as shown in the figure.
Specifically, 
$g_2,g_4,g_5$ are dense grid cells and all the points in them are core points, while $g_1,g_3,g_6$ are sparse grid cells. 
To label the points in these sparse grid cells, Algorithm~\ref{algo:coredete} performs in two MPC rounds. 
Take grid cell $g_3$ which is stored in $M_1$ as an example; 
the neighbour grid cell set of $g_3$ is $\Nei_{g_3}=\{g_1,g_2,g_4\}$, 
and the storage locations of them include: $M_1$ for both $g_1$ and $g_2$, and $M_2, M_3, M_4$ for $g_4$. 
The communications between $M_1$ and $M_1$ (i.e., local computation) and $M_2,M_3,M_4$ are incurred, which are shown by the blue arrows in the figure.
 }
  \label{fig:dbcore}
\end{figure}

\begin{lemmaapprep}
\label{lem:wordsizealgofive}
The above algorithm, i.e., Algorithm~\ref{algo:coredete}, performs two  MPC rounds. 
In each round, the amount of data sent and received by each machine is bounded by $O(s)$.
\end{lemmaapprep}

\begin{proof}
From Algorithm~\ref{algo:coredete}, it is clear that there are only two rounds.

Next, we bound the total number of incoming and outgoing $O(1)$-word messages in the first round.
And then the message number bound of the second round immediately follows from the fact that 
the second round can be regarded as an inverse process of the first one.

Each grid cell can only receive $|g'|<\minPts$ points from its non-empty neighbour grid cells $g'$ and there are at most $O(1)$ non-empty neighbour grid cells for each grid cell. 
Each machine corresponds to at most $O(s)$ grid cells, 
so the total number of incoming $O(1)$-word messages to a machine is bounded by $O(s)$. 
As for the number of outgoing messages, 
Algorithm~\ref{algo:coredete} only sends the point coordinates when the point is in a spare grid cell. 
In each machine, there are at most $O(s)$ such points in total because $\minPts$ is a constant. 
There are only two cases depending on the types of the non-empty neighbour grid cells. 
If a grid cell only has at most $s$ points, this grid cell is a {\em small} grid cell;
otherwise, it is a {\em large} grid cell. 
The storage location of a small grid cell contains at most $O(1)$ machine ID's. 
Sending the point coordinates to small grid cells only takes at most $O(s)$ words in a machine. 
At the same time, there are at most $O(m)$ large grid cells because each machine only contains at most $O(1)$ large grid cells. 
The total number of $O(1)$-word messages received by those large grid cells is thus bounded by $O(m) \in O(s)$. 
Combining the above two cases, we have that the total number of outgoing $O(1)$-word messages of each machine is bounded by $O(s)$.
\end{proof}

\begin{lemma}[Core Point Identification]
\label{lem:corepointidtifi}
Consider a set $P$ of $n$ points; given a radius parameter $\eps \geq 0$ and  a core point threshold, $\minPts \geq 1$,
which is a constant integer, there exists an MPC algorithm with local memory {$\Theta(s) \subseteq  \Theta(n^\alpha)$} per machine, for arbitrary constant $\alpha \in (\frac{1}{2},1)$, which identifies all the core points in {$O(\frac{1}{\alpha})$} rounds. 
\end{lemma}
\begin{proof}
This lemma immediately follows from Lemma~\ref{lem:sendsizebound} and Lemma~\ref{lem:wordsizealgofive}.
\end{proof}

\subsubsection{Compute Primitive Clusters} 
Our algorithm works as follows:
\begin{itemize}
	\item construct an implicit $(d, c)$-grid graph $G' = (V', \mathcal{E}')$:	
	\begin{itemize}
			\item initialize $V' \leftarrow \varnothing$; impose a grid on $P_{\text{core}}$ with side length of $\frac{1}{2\sqrt{d}}\rho\eps$;
			\item for each non-empty grid cell $g$, add the most {\em bottom-left} corner point $t$ of $g$ to $V'$; 
			\item set $\mathcal{E}'$ to consider those edges $(t_1, t_2)$ with $\dist(t_1, t_2) \leq (1 + \frac{1}{2}\rho) \cdot \eps$ only and returns $\dist(t_1, t_2)$ as its weight;
   as a result, $G'$ is $c$-penetration with $c = \frac{(1 + \frac{1}{2} \rho) \cdot \eps}{\frac{1}{2\sqrt{d}}\rho\eps} \in O(1)$;
	\end{itemize}
\item invoke our MPC algorithm in Theorem~\ref{thm:implicit-connectivity} on the implicit $(d,c)$-grid graph $G'=(V', \mathcal{E}')$ to compute an MSF of $G'$ and associate the connected component ID to each point $t \in V'$;  
\item for each point $t \in V'$, 
	assign $t$'s connected component ID to each core point in the grid cell which $t$ represents; 
\item group all the core points in $P_{\text{core}}$ by their connected component ID's; 
\item return each group as a primitive cluster; denote the set of these primitive clusters by $\mathcal{C}$; 
\end{itemize}

\begin{lemmaapprep}\label{lmm:dbscan-correctness}
    The above primitive cluster computing algorithm returns legal $\rho$-approximate primitive clusters in $O(1)$ rounds with high probability. 
\end{lemmaapprep}
\begin{proof}
It can be verified that except the MSF computation, all other steps can be achieved by sorting which can be performed in $O(1)$ MPC rounds.
	Therefore, combining Theorem~\ref{thm:implicit-connectivity}, the total round number follows.

Next, we show the correctness of our algorithm.
In the following proof, for any core point $p \in P_{\text{core}}$, we use notation $p'$ to denote the point $t \in V'$ whose corresponding cell contains $p$.
First of all, we give two crucial observations:
\begin{itemize}
	\item Observation 1: For any $x, y \in P_{\text{core}}$ with $\dist(x,y) \leq \eps$, the edge $(x', y')$ exists in $G'$. 
	\item Observation 2: For any $x, y \in P_{\text{core}}$ with $\dist(x, y) > (1 + \rho) \cdot \eps$, the edge $(x', y')$ must not exist in $G'$. 
\end{itemize}
These two observations follow directly from the fact that $|\dist(x, y) - \dist(x', y')|\leq \frac{1}{2}\rho\eps$ and the definition of $\mathcal{E}'$.

Let $T^*$ denote the EMST of $P_\text{core}$, $T$ denote an MSF of $G'$ computed by our algorithm. 
For any $u,v\in P_\text{core}$, 
the path from $u$ to $v$ on $T^*$ is denoted by $\P^*(u,v)$, and we denote the point sequence on $\P^*(u,v)$ as:
$p_0=u,p_1,p_2,\cdots,p_k=v$
.
Consider the {\em heaviest} (with the largest weight) edge $(x, y)$ on $\P^*(u,v)$. 
We prove the following two claims:
\begin{itemize}
	\item Claim 1: if $\dist(x,y) \leq \eps$, $u'$ and $v'$ must be connected in $G'$ and hence, $u$ and $v$ are in the same primitive cluster in $\mathcal{C}$.
	\item Claim 2: if $\dist(x,y) > (1 + \rho) \cdot \eps$, $u'$ and $v'$ must be {\em disconnected} in $G'$, and hence, $u$ and $v$ are in different primitive clusters in $\mathcal{C}$. 
\end{itemize}

\noindent
\underline{Proof of Claim 1.}
As $(x,y)$ is the heaviest edge on $\P^*(u,v)$, each edge $(p_i, p_{i+1})$ on $\P^*(u,v)$ has weight $\leq \eps$.
By Observation 1, the edge $(p_i', p_{i+1}')$ exists in $G'$ and therefore, $u'$ and $v'$ are connected in $G'$.

\noindent
\underline{Proof of Claim 2.}
First, observe that 
when the heaviest edge $(x,y) \in \P^*(u,v)$ in EMST $T^*$ has weight $\dist(x,y) > (1+\rho) \cdot \eps$, 
any path $\P$ from $u$ to $v$ in the {\em complete graph} $G_{\text{comp}}$ on $P_{\text{core}}$ must have their heaviest edge with weight $> (1 + \rho) \cdot \eps$.
This is because otherwise, $T^*$ will not be the EMST of $P_{\text{core}}$. 
Moreover, this case indeed implies that there does not exist an $\rho$-approximate core graph $G_\rho$ in which $u$ and $v$ are connected.
In other words, $u$ and $v$ are not allowed to be in the same $\rho$-approximate primitive cluster.
On the other hand, by Observation 2, 
for the heaviest edge $(a, b)$ on any path $\P$ from $u$ to $v$ in $G_{\text{comp}}$, 
the edge $(a', b')$ must not exist in $G'$.
And therefore, $u'$ and $v'$ must be disconnected in $G'$.

Putting these two claims together, we know that $\mathcal{C}$ is a legal set of $\rho$-approximate primitive clusters.
This completes the proof of Lemma~\ref{lmm:dbscan-correctness}.
\end{proof}

\section{Conclusion}
\label{sec:conclusion}

In this paper, we study the problems of grid graph connectivity, approximate EMST and approximate DBSCAN in the MPC model with strictly sub-linear local memory space per machine.
We show Las Vegas algorithms (succeeding with high probability), which can be derandomized, for solving these problems in $O(1)$ rounds.
Due to the importance of the problems we studied in this paper, we believe that our $O(1)$-round MPC algorithms and 
then pseudo $s$-separator technique will be of independent interests for solving other related problems in the MPC model.

\newpage
\bibliography{ref}


\end{document}